\documentclass[reqno,11pt]{amsart}
\usepackage{amssymb}
\usepackage{amsmath}
\usepackage{amsthm}
\usepackage{mathtools}
\usepackage{amsfonts}
\usepackage{xcolor}
\usepackage{graphicx}
\usepackage{tikz}
\usepackage{comment}
\usepackage{microtype}
\usepackage{enumitem}
\usepackage[margin=1.5in]{geometry}
\usepackage[alphabetic, initials, nobysame]{amsrefs}
\renewcommand{\MR}[1]{}
\usepackage[
    colorlinks = true, 
    linkcolor={blue},
	citecolor={blue},
	urlcolor={black!30!blue}
]{hyperref}

\title[Alternating Distribution of Resonances]{The Sharp Spectral Transition for Almost Mathieu Operators via Alternating Resonances}

\author[J. \ He]{Jiawei He}
\address{[J. \ He]
   Fujian Key Laboratory of Financial Information Processing, Putian University, Putian, Fujian 351100, P.R. China.
}

\email{\href{mailto:hermit_well@163.com}{hermit\_well@163.com}}

\author[X.\ Wang]{Xueyin Wang}
\address{[X.\ Wang] Department of Mathematics, Texas A\&M University, College Station, TX 77843, USA}
\email{\href{mailto:xueyin@tamu.edu}{xueyin@tamu.edu}}

\theoremstyle{plain}
\newtheorem{theorem}{Theorem}[section]
\newtheorem{conjecture}{Conjecture}[section]

\newtheorem{lemma}[theorem]{Lemma}
\newtheorem{proposition}[theorem]{Proposition}

\theoremstyle{definition}

\numberwithin{equation}{section}
\DeclareMathOperator{\dist}{dist}
\DeclareMathOperator{\Lag}{Lag}

\begin{document}

\begin{abstract}
We prove that for any given frequency resonance exponent and phase resonance exponent, there exist a frequency and a phase exactly realizing these values such that the Almost Mathieu operator exhibits Anderson localization for
$\ln|\lambda|>\max\{\beta(\alpha),\delta(\alpha,\theta)\}$. This resolves the conjecture in \cite{MR4686650}.

\end{abstract}

\maketitle	

\section{Introduction}
This paper studies the almost Mathieu operator given by
\begin{equation*}
    [H_{\lambda,\alpha,\theta} \phi](n) = \phi(n+1) + \phi(n-1) + 2\lambda \cos 2\pi (\theta+n\alpha)\phi(n),
\end{equation*}
where $\theta \in \mathbb{R}/\mathbb{Z}$ is the phase, $\alpha \in \mathbb{R}\setminus \mathbb{Q}$ is the irrational frequency, and $\lambda \in \mathbb{R}$ is the coupling constant. The almost Mathieu operator is one of the central models in the spectral theory of quasi-periodic Schr\"odinger operators \cites{MR4567742,MR4840232}. 
 The spectral type is completely understood in the subcritical and critical regimes, while in the supercritical regime it is strongly affected by arithmetic resonances.

There are two natural types of resonances in this setting. The first one is the frequency resonance \cites{MR634437,MR458247}, which is quantified by
\begin{equation*}
    \beta(\alpha) \coloneqq \limsup_{|k|\to \infty} -\frac{\ln \|k\alpha\|_{\mathbb{R}/\mathbb{Z}}}{|k|},
\end{equation*}
where $\|x\|_{\mathbb R/\mathbb Z}=\inf_{m\in\mathbb Z}|x-m|$. The second one is the phase resonance,  which occurs for even functions \cite{MR1298948}. It is quantified by
\begin{equation*}
    \delta(\alpha,\theta) \coloneqq \limsup_{|k|\to \infty} -\frac{\ln \|2\theta+k\alpha\|_{\mathbb{R}/\mathbb{Z}}}{|k|}.
\end{equation*}

Let us recall the known spectral picture. In the subcritical case ($|\lambda|<1$), following numerous contributions \cites{avila2008absolutely,MR2390290,MR2578605,MR1459260,MR1740982}, it is well established that $H_{\lambda,\alpha,\theta}$ exhibits purely absolutely continuous spectrum for all  $\alpha$ and all phases $\theta$. For the critical case ($|\lambda|=1$), recent works \cites{MR4322159,MR3698344,MR3801496} imply that $H_{\lambda,\alpha,\theta}$ exhibits purely singular continuous spectrum for all irrational $\alpha$ and all phases $\theta$. 

The supercritical regime ($|\lambda|>1$) is more delicate, as the spectral properties in this regime depend strongly on resonances. More precisely, on the singular continuous side, if $0<\ln |\lambda|<\beta(\alpha)$, $H_{\lambda,\alpha,\theta}$ has purely singular continuous spectrum for all $\theta$ \cites{MR3779957,MR3707287}; if $0<\ln |\lambda|<\delta(\alpha,\theta)$, $H_{\lambda,\alpha,\theta}$ has purely singular continuous spectrum for all $\alpha$ \cite{MR4756946}. Combining these results \cites{MR3779957,MR3707287,MR4756946} implies that $H_{\lambda,\alpha,\theta}$ possesses a purely singular continuous spectrum whenever $0 < \ln |\lambda| < \max\{\beta(\alpha),\delta(\alpha,\theta)\}$.

On the Anderson localization side, several sharp results are known in special regimes: 
(1) $\ln|\lambda|>\beta(\alpha)$ and $\delta(\alpha,\theta)=0$, see \cites{MR3779957,MR3707287}; (2) $\ln |\lambda|>\delta(\alpha,\theta)$ and $\beta(\alpha)=0$, see \cite{MR4756946}; (3) $\ln |\lambda|>2\beta(\alpha)$ and $2\theta\in \alpha\mathbb{Z}+ \mathbb{Z}$, see \cite{MR4950602}. Since $\delta(\alpha,\theta)=\beta(\alpha)$ for $2\theta\in \alpha\mathbb{Z}+ \mathbb{Z}$, this suggests that $\ln|\lambda|=\beta(\alpha)+\delta(\alpha,\theta)$ might be the spectral transition line. It was natural to ask whether this sum always governs the competition between localization and resonances.

However, the recent discovery in \cite{MR4686650} demonstrates that this is not the complete picture. Not only the strengths of frequency and phase resonances, measured by $\beta(\alpha)$ and $\delta(\alpha,\theta)$, but also the distribution of the two families of resonances matters. Specifically, the author constructed the irrational $\alpha$ and the phase $\theta$ with $\delta(\alpha,\theta)>100\beta(\alpha)$ such that Anderson localization holds provided $\ln |\lambda|>\delta(\alpha,\theta)$. This leads to the following conjecture formulated in \cite{MR4686650}.

\begin{conjecture}\label{mainconjec}
    For any $0<\mu<\infty$ and $0<\kappa<\infty$, there exist $\alpha$ and $\theta$ with $\beta(\alpha)=\mu$ and $\delta(\alpha,\theta)=\kappa$ such that the following is not true: $H_{\lambda,\alpha,\theta}$ has Anderson localization for $\ln|\lambda|>\kappa+\mu$ and has purely singular continuous (sc) spectrum for $0<\ln|\lambda|<\kappa+\mu$.
\end{conjecture}

The main result of this paper confirms Conjecture \ref{mainconjec} in a sharp form. We construct frequencies and phases for which the transition is governed not by the sum $\mu+\kappa$, but by the larger of the two resonance strengths.

\begin{theorem}\label{mainthm}
    For any $0<\mu<\infty$ and $0<\kappa<\infty$, there exist $\alpha$ and $\theta$ with $\beta(\alpha)=\mu$ and $\delta(\alpha,\theta)=\kappa$ such that the following hold:
    \begin{enumerate}
        \item \label{item:AL}$H_{\lambda,\alpha,\theta}$ has Anderson localization for $\ln|\lambda|>\max \{\kappa,\mu\}$;
        \item \label{item:SC}$H_{\lambda,\alpha,\theta}$ has purely sc spectrum for $0<\ln|\lambda|<\max \{\kappa,\mu\}$.
    \end{enumerate}
\end{theorem}

The singular continuous part, namely \hyperref[item:SC]{Theorem~\ref*{mainthm}\,\textup{(\ref*{item:SC})}}, is not new. Indeed, for the constructed pair $(\alpha,\theta)$, one has
\begin{equation*}
    \max\{\beta(\alpha),\delta(\alpha,\theta)\}
    =
    \max\{\mu,\kappa\}.
\end{equation*}
Consequently, this follows directly from the aforementioned results \cites{MR3779957,MR3707287,MR4756946}. Thus the main contribution of this paper is the Anderson localization statement in \hyperref[item:AL]{Theorem~\ref*{mainthm}\,\textup{(\ref*{item:AL})}}.

We now outline the main idea of the proof. The starting point is the observation that the arithmetic quantities $\beta(\alpha)$ and $\delta(\alpha,\theta)$ measure only the strengths of the frequency and phase resonances, but not their locations. Our construction exploits this additional freedom. We arrange the distributions of frequency and phase resonances so that their strongest contributions occur at different sites. Consequently, the Lyapunov exponent does not need to dominate the sum of the two resonance strengths. Instead, at each scale it only needs to dominate the resonance which is actually present at the relevant site. This mechanism gives rise to the threshold condition
\begin{equation*}
    \ln|\lambda| \sim \max\{\beta(\alpha),\delta(\alpha,\theta)\}.
\end{equation*}

The construction is based on an alternating choice of the continued fraction denominators $q_{n}$ depending on whether $n$ is odd or even. The essential geometric feature of the construction is the separation of resonant sites. At scale $q_n$, frequency resonances occur near the sites $\ell q_{n}$, $\ell\in\mathbb{Z}$. When $n=2j-1$ is odd, a phase resonance occurs near the site $(\ell+\eta) q_n$ for some fixed $0<\eta<1$. The point $(\ell+\eta) q_n$ lies strictly between the two neighboring frequency-resonant sites $\ell q_n$ and $(\ell+1)q_n$. Hence, a phase resonance and a frequency resonance do not occur at the same site. When $n=2j$ is even, the construction keeps the phase resonance away from the relevant scale, and the only effective resonance is the frequency resonance near $\ell q_n$.

This paper is organized as follows. In Section \ref{sec:prelim}, we introduce some necessary notations and lemmata. Section \ref{construction} is devoted to the explicit construction of the resonant frequency and phase. We then establish localization via a site dichotomy: Section~\ref{sec:Non-resonant} addresses non-resonant sites, whereas Section~\ref{sec:Resonant} analyzes resonant sites by distinction of parity, encompassing the estimation of $r_{\ell+\eta}$ and $r_{\ell}$ for $n=2j-1$ and the analysis for $n=2j$.

\section{Preliminaries}\label{sec:prelim}

\subsection{Continued fraction}
Let $\alpha\in(0,1)\setminus\mathbb Q$. We write its continued fraction expansion as
\begin{equation*}
    \alpha=[0;a_1,a_2,\ldots]\coloneqq\cfrac{1}{a_1+\cfrac{1}{a_2+\cdots}}.
\end{equation*}
Let $p_{-1}=1$, $p_0=0$, $q_{-1}=0$, and $q_0=1$. For $n\geqslant 0$, define
\begin{equation*}
    p_{n+1}=a_{n+1}p_n+p_{n-1},
    \quad
    q_{n+1}=a_{n+1}q_n+q_{n-1}.
\end{equation*}
It is standard that $(-1)^{n}(q_{n}\alpha-p_{n})>0$. Moreover, $p_n/q_n$ is the $n$-th convergent of $\alpha$, 
\begin{equation*}
    \frac{p_n}{q_n}=[0;a_1,\ldots,a_n],\quad n\geqslant 1.
\end{equation*}
For $n\geqslant 1$, we shall use
\begin{equation*}
    \frac{1}{2q_{n+1}}\leqslant\|q_n\alpha\|_{\mathbb{R}/\mathbb{Z}}=|q_n\alpha-p_n|\leqslant\frac{1}{q_{n+1}},
\end{equation*}
and
\begin{equation*}
    \|k\alpha\|_{\mathbb{R}/\mathbb{Z}}\geqslant\|q_n\alpha\|_{\mathbb{R}/\mathbb{Z}},\quad \text{for}\ 1\leqslant |k|<q_{n+1}.
\end{equation*}

\subsection{Lyapunov exponent}
We say that $\phi$ is a generalized eigenfunction corresponding to the generalized eigenvalue $E$ if
\begin{equation}\label{eigen1}
    H_{\lambda,\alpha,\theta}\phi=E\phi\quad \text{and}\quad |\phi(n)|\lesssim |n|+1.
\end{equation}
By Shnol's theorem, in order to prove Anderson localization, it is enough to show that every generalized eigenfunction decays exponentially. From now on, we always assume that $\phi$ is a generalized eigenfunction of $H_{\lambda,\alpha,\theta}$. Since $|\phi(0)|+|\phi(1)|\neq 0$, for simplicity, we normalize
\begin{equation}\label{normalization}
    \phi(0)=1.
\end{equation}

Throughout the paper, we denote \begin{equation*}
    A(\theta)=\begin{pmatrix}
    E-2\lambda \cos 2\pi \theta&-1\\1&0
\end{pmatrix}\in \mathrm{SL}(2,\mathbb{R}).
\end{equation*}
For $n\geqslant 1$, we define the $n$-step transfer matrix
\begin{equation*}
    A_{n}(\theta)=A(\theta+(n-1)\alpha)\cdots A(\theta+\alpha)A(\theta)
\end{equation*}
and $A_{-n}=A_{n}(\theta-n\alpha)^{-1}$. 
Any formal eigenfunction of $H\phi=E\phi$ satisfies
\begin{equation}\label{preeigen}
    \begin{pmatrix}
        \phi(n)\\
        \phi(n-1)
    \end{pmatrix}=
    A_{n-m}(\theta+m\alpha)
    \begin{pmatrix}
        \phi(m)\\
        \phi(m-1)
    \end{pmatrix}.
\end{equation}
The Lyapunov exponent is given by
\begin{equation*}
    L(E)=\lim _{n \to \infty} \frac{1}{n} \int_{\mathbb{R}/\mathbb{Z}} \ln \|A_{n}(\theta)\| \ \mathrm{d} \theta.
\end{equation*}
Moreover, we have expression of $L(E)$ for $E\in \Sigma_{\lambda,\alpha}$.
\begin{lemma}[\cites{MR1933451}]
    For $E\in \Sigma_{\lambda,\alpha}$, we have $L(E)=\max \{0,\ln |\lambda|\}$.
\end{lemma}
Since the irrational rotation is unique ergodic, Furman's theorem implies that for any $\varepsilon>0$, there exists $N=N(\varepsilon,E)$ such that
\begin{equation}\label{prefur}
    \|A_{n}(\theta)\|\leqslant e^{(L(E)+\varepsilon)|n|},\quad \text{for}\ |n|>N.
\end{equation}
By \eqref{preeigen} and \eqref{prefur}, for sufficiently large $|m-n|$,
\begin{equation}\label{transfermatrix}
    \bigg\|\begin{pmatrix}
        \phi(n)\\
        \phi(n-1)
    \end{pmatrix}\bigg\|\leqslant e^{(L(E)+\varepsilon)|m-n|} \bigg\|\begin{pmatrix}
        \phi(m)\\
        \phi(m-1)
    \end{pmatrix}\bigg\|.
\end{equation}

Let us denote
\begin{equation*}
    P_{n}(\theta)=\det (R_{[0, n-1]}(E-H_{\lambda, \alpha, \theta})R_{[0, n-1]} ),
\end{equation*}
where $R_{[x_{1}, x_{2}]}$ is the Dirichlet restriction to $[x_{1}, x_{2}] \subseteq \mathbb{Z}$ and we define $P_{0}(\theta)=1, P_{-1}(\theta)=0$ for convention. Let $P_{[x_{1},x_{2}]}(\theta)=P_{n}(\theta+x_{1}\alpha)$ where $n=x_{2}-x_{1}+1$. A direct computation shows that
\begin{equation*}
A_n(\theta)=\begin{pmatrix}
      P_{n}(\theta)&-P_{n-1}(\theta+\alpha) \\
      P_{n-1}(\theta) & -P_{n-2}(\theta+\alpha)
  \end{pmatrix}.
\end{equation*}
Then by \eqref{prefur}, for sufficiently large $|x_{2}-x_{1}+1|$,
\begin{equation}\label{upperbound}
    |P_{[x_{1},x_{2}]}(\theta)| \leqslant e^{(L(E)+\varepsilon)|x_{2}-x_{1}+1|}.
\end{equation}

\subsection{Block expansion}
For any interval $[x_{1},x_{2}] \subseteq \mathbb{Z}$, define the Green's function
\begin{equation*}
    G_{[x_{1},x_{2}]}=(R_{[x_{1},x_{2}]}(H_{\lambda, \alpha, \theta}-E) R_{[x_{1},x_{2}]})^{-1}.
\end{equation*}
For any $y\in [x_{1},x_{2}]$, the eigenfunction satisfies
\begin{equation}\label{preblock}
    \phi(y)=-G_{[x_{1}, x_{2}]}(y,x_{1}) \phi(x_{1}')-G_{[x_{1}, x_{2}]}(y,x_{2}) \phi(x_{2}'),
\end{equation}
where $x_{1}'=x_{1}-1$ and $x_{2}'=x_{2}+1$. By Cramer's rule, for any $y\in [x_{1},x_{2}]$, 
\begin{equation}\label{pregreen}
    |G_{[x_{1},x_{2}]}(x_{1},y)|=\frac{|P_{[y+1, x_{2}]}|}{|P_{[x_{1},x_{2}]}|},\quad |G_{[x_{1},x_{2}]}(y,x_{2})|=\frac{|P_{[x_{1}, y-1]}|}{|P_{[x_{1},x_{2}]}|}.
\end{equation}
Combine \eqref{preblock} with \eqref{pregreen}, one has
\begin{equation*}
    |\phi(y)|\leqslant \frac{|P_{[y+1, x_{2}]}|}{|P_{[x_{1},x_{2}]}|} |\phi(x_{1}')| +\frac{|P_{[x_{1}, y-1]}|}{|P_{[x_{1},x_{2}]}|} |\phi(x_{2}')|.
\end{equation*}

Given $\{\theta_{1},\cdots,\theta_{n+1}\}$, the Lagrange interpolation terms $\Lag_{m}, m=1,\cdots,n+1$ are defined by
\begin{equation*}
    \Lag_{m}=\ln \max_{a\in\mathbb{T}} \prod_{\ell=1,\ell\neq m}^{n+1} \frac{|\cos 2\pi a-\cos 2\pi\theta_{\ell}|}{|\cos 2\pi \theta_{m}-\cos 2\pi\theta_{\ell}|}.
\end{equation*}
We have the following estimate for the determinant.
\begin{lemma}[\cite{MR4108908}]\label{lowlag}
    Let $\{\theta_{1},\cdots, \theta_{n+1}\} \subseteq \mathbb{T}$. There exists $1\leqslant m\leqslant n+1$ such that
    \begin{equation*}
        \bigg|P_{n}\bigg(\theta_{m}-\frac{n-1}{2}\alpha\bigg)\bigg| \geqslant \frac{e^{n L(E)-\Lag_{m}}}{n+1}.
    \end{equation*}
\end{lemma}

The following block expansion estimate will be used in resonant sites.
\begin{lemma}[\cite{MR4686650}]\label{block}
    Let $[x_{1},x_{2}]\subseteq \mathbb{Z}$ with $x_{2}-x_{1}+1=n$. Assume that $|P_{[x_{1},x_{2}]}|\geqslant e^{L(E)|x_{2}-x_{1}|-\xi}$. Then for any $\varepsilon>0$ and sufficiently large $n$,
    \begin{equation*}
        |\phi(y)|\leqslant e^{\xi} \sum_{i=1,2}|\phi(x_{i}')| e^{-L(E)|y-x_{i}|+\varepsilon n}.
    \end{equation*}
\end{lemma}

For the non-resonant sites, we use following estimates. Recall that for a given $\tau > 0$, a point $y \in \mathbb{Z}$ is called $(\tau, t)$-regular if there exists an interval $[x_1, x_2]$ containing $y$, where $x_2 = x_1 + t - 1$, such that
\begin{equation*}
	|G_{[x_1, x_2]}(y, x_i)| < e^{-\tau |y - x_i|} \quad \text{and} \quad |y - x_i| \geqslant \frac{t}{40} \quad \text{for } i = 1, 2.
\end{equation*}
\begin{theorem}[\cite{MR4756946}]\label{block1}
    Suppose $y_1, y_2 \in \mathbb{Z}$ are such that $y_2 - y_1 = K$. Suppose there exists some $\tau > 0$ such that for any $y \in [y_1 + \gamma K, y_2 - \gamma K]$, $y$ is $(\tau, t)$-regular for some $\frac{\gamma K}{20} < t \leqslant \frac{1}{2}\min\{|y - y_1|, |y - y_2|\}$. 
    Set $R_{y}=\sup_{|\sigma|\leqslant 10\gamma}|\phi(y+\sigma K)|$. Then for large enough $K$ (depending on $\tau$ and $\gamma$),
    \begin{equation*}
        R_y \leqslant \max \big\{ R_{y_1} e^{-\tau(|y - y_1| - 3\gamma K)},\, R_{y_2} e^{-\tau(|y - y_2| - 3\gamma K)} \big\}.
    \end{equation*}
for all $y \in [y_1 + 10\gamma K, y_2 - 10\gamma K]$.
\end{theorem}

In the rest of the paper, $C>0$ denotes a large constant independent of $n$. Its value may change from line to line. We always take $n$ sufficiently large, depending on the fixed parameters and constants. We use $\{x\}$ to represent the fractional part of $x\in\mathbb{R}$.

\section{Construction of frequency and phase}\label{construction}
Fix any $0 < \mu < \infty$ and $0 < \kappa < \infty$. We first construct $\alpha$ and $\theta$ such that $\beta(\alpha)=\mu$ and $\delta(\alpha,\theta)=\kappa$.
\subsection{Construction of the frequency}
Let 
\begin{equation}\label{eta}
    0<\eta<\min \bigg\{10^{-2},\frac{\mu}{2\kappa}\bigg\}.
\end{equation}
Choose $j_0$ sufficiently large. We use the standard notation for
continued fractions: $q_{-1}=0$, $q_{0}=1$, and 
\begin{equation*}
    q_{n+1}=a_{n+1}q_{n}+q_{n-1}.
\end{equation*}
Choose $a_{1}=2$. For $1\leqslant n\leqslant j_{0}$, define $a_{n+1}=q_{n}^{10}$. And for $n>j_{0}$, define
\begin{equation}\label{an}
    a_{n+1} = \begin{cases}
        \left\lfloor \dfrac{e^{\kappa\eta q_n}}{q_n} \right\rfloor, & \text{if } n = 2j-1, \\[2ex]
        \left\lfloor \dfrac{e^{\mu q_n}}{q_n} \right\rfloor, & \text{if } n = 2j.
    \end{cases}
\end{equation}

Define the frequency $\alpha$ via 
\begin{equation*}
    \alpha\coloneq\cfrac{1}{a_{1}+\cfrac{1}{a_{2}+\cdots}}\in\mathbb{R}\setminus\mathbb{Q}.
\end{equation*}
By \eqref{an}, as $j\to\infty$,
\begin{equation*}
    q_{2j}=e^{\kappa\eta q_{2j-1}}(1+o(1)),\quad
    q_{2j+1}=e^{\mu q_{2j}}(1+o(1)).
\end{equation*}
Hence, as $j\to\infty$,
\begin{equation*}
    \frac{\ln q_{2j}}{q_{2j-1}}\to \kappa\eta,
    \quad
    \frac{\ln q_{2j+1}}{q_{2j}}\to \mu.
\end{equation*}
Since $\kappa\eta<\mu$ by \eqref{eta}, we obtain
\begin{equation*}
    \beta(\alpha)=\limsup_{n\to\infty}\frac{\ln q_{n+1}}{q_{n}}=\mu.
\end{equation*}

\subsection{Construction of the phase}
Since the irrational rotation by $\alpha$ is minimal, we can choose
a sufficiently large positive integer $K$ such that
\begin{equation*}
    \{-K\alpha\}\in\bigg[\frac{1}{10},\frac{1}{5}\bigg].
\end{equation*}
By choosing $j_0$ sufficiently large, we have that for all $j \ge j_0$,
\begin{equation*}
    \eta \frac{q_{2j+1}}{q_{2j}} \ge 2.
\end{equation*}
Set $k_{j_0}=K$. For $j<j_0$, set $k_j=0$. For $j\geqslant j_{0}$, define
\begin{equation}\label{defin_k}
    k_{j+1}=k_{j}+\bigg\lfloor\eta \frac{q_{2j+1}}{q_{2j}}\bigg\rfloor q_{2j}.
\end{equation}
And for $j\geqslant j_{0}$, we define the interval

\begin{equation}\label{Ij}
    I_{j}=\bigg\{\theta\in [0,1/2]: -\frac{10\eta}{q_{2j}}\leqslant 2\theta-\{-k_{j}\alpha\} \leqslant -\frac{\eta}{10 q_{2j}}\bigg\}.
\end{equation}

\begin{lemma}\label{closeint}
    There exists $\theta= \bigcap_{j\geqslant j_{0}}I_{j}$.
\end{lemma}
\begin{proof}
Let
\begin{equation*}
    m_{j}=\bigg\lfloor\eta \frac{q_{2j+1}}{q_{2j}}\bigg\rfloor,
\end{equation*}
then  by the definition of $j_0$, we have
\begin{equation}\label{dkj}
    \frac{\eta q_{2j+1}}{2} \leqslant k_{j+1}-k_{j}=m_{j}q_{2j}\leqslant \eta q_{2j+1}.
\end{equation}
For even labels, by standard Diophantine approximation, we have $q_{2j}\alpha-p_{2j}>0$. Moreover,
\begin{equation*}
    \frac{1}{2q_{2j+1}}< \|q_{2j}\alpha\|_{\mathbb{R}/\mathbb{Z}}< \frac{1}{q_{2j+1}}.
\end{equation*}
Since $q_{2j+1}\gg q_{2j}$, it follows that
\begin{equation*}
    \frac{\eta}{4q_{2j}}\leqslant m_j\|q_{2j}\alpha\|_{\mathbb{R}/\mathbb{Z}}\leqslant \frac{\eta}{q_{2j}}.
\end{equation*}
By \eqref{dkj}, we have
\begin{equation}\label{dkjalpha}
    -\frac{\eta}{q_{2j}}\leqslant\{-k_{j+1}\alpha\}-\{-k_{j}\alpha\}\leqslant -\frac{\eta}{4q_{2j}}.
\end{equation}

Since $\{-K\alpha\}\in [\frac{1}{10},\frac{1}{5}]$ and
\begin{equation*}
    \sum_{j\geqslant j_{0}} \frac{\eta}{q_{2j}}\leqslant \frac{1}{20},
\end{equation*}
we obtain that
\begin{equation*}
    \{-k_{j}\alpha\}\in \bigg[\frac{1}{20},\frac{1}{5}\bigg]\quad \text{for any}\ j\geqslant j_{0}.
\end{equation*}
Write
\begin{equation*}
    I_{j}=\bigg[\frac{A_{j}}{2}, \frac{B_{j}}{2}\bigg],
\end{equation*}
where
\begin{equation*}
    A_{j}=\{-k_{j}\alpha\}-\frac{10\eta}{q_{2j}},\quad B_{j}=\{-k_{j}\alpha\}-\frac{\eta}{10 q_{2j}}.
\end{equation*}
On the one hand, by \eqref{dkjalpha} and $q_{2j+2}\gg q_{2j}$,
\begin{equation*}
    \begin{split}
        A_{j+1}-A_{j}&=(\{-k_{j+1}\alpha\}-\{-k_{j}\alpha\})+\frac{10\eta}{q_{2j}}-\frac{10\eta}{q_{2j+2}}\\
        &\geqslant -\frac{\eta}{q_{2j}}+\frac{10\eta}{q_{2j}}-\frac{10\eta}{q_{2j+2}}\\
        &\geqslant\frac{8\eta}{q_{2j}}>0.
    \end{split}
\end{equation*}
On the other hand, by \eqref{dkjalpha} 
\begin{equation*}
    \begin{split}
        B_{j+1}-B_{j}&=(\{-k_{j+1}\alpha\}-\{-k_{j}\alpha\})+\frac{\eta}{10q_{2j}}-\frac{\eta}{10q_{2j+2}}\\
        &\leqslant -\frac{\eta}{4q_{2j}}+\frac{\eta}{10q_{2j}}\\
        &\leqslant -\frac{\eta}{10q_{2j}}<0.
    \end{split}
\end{equation*}
Thus $I_{j+1}\subseteq I_{j}$ for any $j\geqslant j_{0}$. As $\{I_j\}_{j\geqslant j_0}$ is a decreasing sequence of nonempty
closed intervals and $|I_{j}|\to 0$ as $j\to\infty$, there exists exactly one point $\theta$ such that
\begin{equation*}
    \theta\coloneq\bigcap_{j\geqslant {j_0}} I_j .   
\end{equation*}
This finishes the proof.
\end{proof}

\begin{lemma}\label{appkj} 
    For sufficiently large $j_{0}$,
    \begin{equation*}
        |k_{j}-\eta q_{2j-1}|\leqslant 100 q_{2j-2},\quad \text{for all}\ j> j_{0}.
    \end{equation*}
\end{lemma}
\begin{proof}
For $\ell\geqslant j_{0}$, we set
\begin{equation*}
    \rho_{\ell}\coloneqq k_{\ell+1}-k_{\ell}-\eta q_{2\ell +1}.
\end{equation*}
By \eqref{defin_k}, we have
\begin{equation*}
    -q_{2\ell}\leqslant \rho_{\ell}\leqslant 0.
\end{equation*}
Thus for $j> j_{0}$,
\begin{equation*}
    k_{j}=k_{j_{0}}+ \sum_{\ell=j_{0}}^{j-1}(k_{\ell+1}-k_{\ell}) =k_{j_{0}}+\sum_{\ell=j_{0}}^{j-1} (\eta q_{2\ell+1}+\rho_{\ell}).
\end{equation*}
Therefore,
\begin{equation*}
    \begin{split}
        |k_{j}-\eta q_{2j-1}|&\leqslant k_{j_0}+ \eta\sum_{\ell=j_0}^{j-2}q_{2\ell+1}+ \sum_{\ell=j_0}^{j-1}q_{2\ell}\\
        &\leqslant q_{2j-2}+2\eta q_{2j-2}+ 2q_{2j-2}<100 q_{2j-2}.
    \end{split}
\end{equation*}
This finishes the proof.
\end{proof}

\begin{lemma}
    Let $\theta$ be in \hyperref[closeint]{Lemma~\ref*{closeint}}. Then
    \begin{equation*}
        \lim_{j\to \infty} -\frac{\ln \|2\theta+k_{j}\alpha\|}{k_{j}}=\kappa=\delta(\alpha,\theta).
    \end{equation*}
\end{lemma}
\begin{proof}
We first prove the first equality. By  \eqref{Ij}, we have
\begin{equation*}
    \frac{\eta}{10q_{2j}}\leqslant \|2\theta+k_{j}\alpha\|_{\mathbb{R}/\mathbb{Z}}\leqslant \frac{10\eta}{q_{2j}}.
\end{equation*}
By \eqref{an}, we have
\begin{equation}\label{qs}
    \ln q_{2j}=\kappa\eta q_{2j-1} +o(q_{2j-1}).
\end{equation}
By \hyperref[appkj]{Lemma~\ref*{appkj}}, we have
\begin{equation}\label{ks}
    k_{j}=\eta q_{2j-1}+ O(q_{2j-2}).
\end{equation}
Combine \eqref{qs} with \eqref{ks}, one has
\begin{equation}\label{firsthalf}
    \lim_{j\to \infty}-\frac{\ln \|2\theta+k_{j}\alpha\|_{\mathbb{R}/\mathbb{Z}}}{k_{j}}=\kappa.
\end{equation}

For the second equality, by \eqref{firsthalf},
\begin{equation*}
	\delta(\alpha,\theta)=\limsup_{|k|\rightarrow \infty}-\frac{\ln \|2\theta+k\alpha\|}{|k|} \geqslant \limsup_{j\rightarrow \infty}- \frac{\ln \|2\theta+k_{j}\alpha\|}{k_{j}}=\kappa.
\end{equation*}
It remains to prove the opposite inequality. Choose $j$ such that $k_{j}\leqslant |k|<k_{j+1}$.

If $|k|\geqslant 10^{-3}\eta q_{2j+1}$, then $0<|k-k_{j+1}|<q_{2j+1}$ and by \eqref{Ij},
\begin{equation*}
	\begin{split}
		\|2\theta+k\alpha\|&\geqslant \|(k-k_{j+1})\alpha\|-\|2\theta+k_{j+1}\alpha\|\\
		&\geqslant \|q_{2j}\alpha\|-\|2\theta+k_{j+1}\alpha\|\\
		&\geqslant \frac{1}{2q_{2j+1}} -  \frac{10\eta}{q_{2j+2}}\\
		&>\frac{1}{4q_{2j+1}}.
	\end{split}
\end{equation*}
Therefore,
\begin{equation}\label{kkj1}
	\begin{split}
		-\frac{\ln \|2\theta+k\alpha\|}{|k|} &\leqslant \frac{\ln 4q_{2j+1}}{10^{-3}\eta q_{2j+1}}\to 0.
	\end{split}
\end{equation}

If $k_{j}\leqslant |k|<10^{-3}\eta q_{2j+1}$, then we write
\begin{equation*}
	k=k_{j}+l_{1}q_{2j}+l_{2},
\end{equation*}
where $|l_{1}|\leqslant \frac{2\cdot 10^{-3}\eta q_{2j+1}}{q_{2j}}$ and $0\leqslant l_{2}<q_{2j}$. If $l_{2}=0$, then
\begin{equation*}
	\begin{split}
		\|2\theta+k\alpha\|&\geqslant \|2\theta+k_{j}\alpha\|-\|l_{1}q_{2j}\alpha\| \\
		&\geqslant \frac{\eta}{10 q_{2j}}- \frac{2\cdot 10^{-3}\eta q_{2j+1}}{q_{2j}} \frac{1}{q_{2j+1}} \\
		&>\frac{\eta}{20 q_{2j}}.
	\end{split}
\end{equation*}
If $l_{2}\neq 0$, then $1\leqslant l_{2}<q_{2j}$. Hence
\begin{equation*}
	\begin{split}
		\|2\theta+k\alpha\|&\geqslant \|l_{2}\alpha\|-\|2\theta+k_{j}\alpha\|- |l_{1}|\|q_{2j}\alpha\|\\
		&\geqslant \frac{1}{2q_{2j}}-\frac{10\eta}{q_{2j}}-\frac{2\cdot 10^{-3}\eta q_{2j+1}}{q_{2j}} \frac{1}{q_{2j+1}}\\
		&\geqslant \frac{1}{4q_{2j}},
	\end{split}
\end{equation*}
where we use $\eta\leqslant 10^{-2}$. 

Therefore, in both cases, for $k_{j}\leqslant |k|<10^{-3}\eta q_{2j+1}$, by \eqref{an} we have
\begin{equation}\label{kkj2}
	-\frac{\ln \|2\theta+k\alpha\|}{|k|} \leqslant \frac{\ln cq_{2j}}{k_{j}} \to \kappa.
\end{equation}
Combine \eqref{kkj1} with \eqref{kkj2},
\begin{equation*}
	\delta(\alpha,\theta)=\limsup_{|k|\to \infty}-\frac{\ln \|2\theta+k\alpha\|}{|k|} \leqslant \kappa.
\end{equation*}
This finishes the proof.
\end{proof}

\section{Localization: Non-resonant sites}
\label{sec:Non-resonant}
Denote $L=\ln |\lambda|$. Let $0<\varepsilon<10^{-5}\eta$ be sufficiently small. Set $b_{n}=10^{-7}\eta q_{n}$. For $\ell\in\mathbb{Z}$, define
\begin{equation*}
    r_{\ell}=r_{\ell}^{\varepsilon,n}\coloneq \sup_{|x-\ell q_{n}|\leqslant 10\varepsilon q_{n}} |\phi(x)|.
\end{equation*}
When $n=2j-1$, define also
\begin{equation*}
    r_{\ell+\eta}=r_{\ell+\eta}^{\varepsilon,n}\coloneq \sup_{|x-(\ell q_{n}+k_{j})|\leqslant 10\varepsilon q_{n}} |\phi(x)|.
\end{equation*}

\begin{theorem}\label{nonres}
Assume $|\lambda|>1$. Let $\ell$ be such that \begin{equation*}
    0\leqslant |\ell|\leqslant 50 \frac{b_{n+1}}{q_{n}}.
\end{equation*}
Then, for sufficiently large $n$, the following hold.

\noindent (1) If $n=2j-1$, then for any 
    \begin{equation*}
        y\in [\ell q_{n}+10\varepsilon q_{n}, \ell q_{n}+k_{j}-10\varepsilon q_{n}],
    \end{equation*}
    we have
    \begin{equation*}
        |\phi(y)|\leqslant r_{\ell}e^{-(L-\varepsilon)(|y-\ell q_{n}|-3\varepsilon q_{n})} +r_{\ell+\eta} e^{-(L-\varepsilon)(|\ell q_{n}+k_{j}-y|-3\varepsilon q_{n})},
    \end{equation*}
    and for any 
    \begin{equation*}
        y\in [\ell q_{n}+k_{j}+10\varepsilon q_{n}, (\ell+1)q_{n}-10\varepsilon q_{n}],
    \end{equation*}
    we have
    \begin{equation*}
        |\phi(y)|\leqslant r_{\ell+\eta} e^{-(L-\varepsilon)(|\ell q_{n}+k_{j}-y|-3\varepsilon q_{n})}+r_{\ell+1}e^{-(L-\varepsilon)(|(\ell+1) q_{n}-y|-3\varepsilon q_{n})}.
    \end{equation*}

\noindent (2) If $n=2j$, then for any 
    \begin{equation*}
        y\in [\ell q_{n}+10\varepsilon q_{n}, (\ell+1)q_{n}-10\varepsilon q_{n}],
    \end{equation*}
    we have
    \begin{equation*}
        |\phi(y)|\leqslant r_{\ell}e^{-(L-\varepsilon)(|y-\ell q_{n}|-3\varepsilon q_{n})}+r_{\ell+1}e^{-(L-\varepsilon)(|(\ell+1)q_{n}-y|-3\varepsilon q_{n})}.
    \end{equation*}
\end{theorem}
\begin{proof}
We first consider the case $n=2j-1$.  For any $p \in [\ell q_n, (\ell+1) q_n ]$ satisfying $|p-\ell q_n| \geqslant \varepsilon q_n$, $|p-(\ell+1) q_n| \geqslant \varepsilon q_n$, and $|p-\ell q_n-k_j| \geqslant \varepsilon q_n$, we let
\begin{equation*}
    d_{p}=\frac{1}{100}\min \{\dist(p,q_{n}\mathbb{Z}), \dist(p-k_{j}, q_{n}\mathbb{Z}), 10^{-5}\eta q_{n} \}.
\end{equation*}
Let $n_{0}$ be the smallest integer such that $2q_{n-n_{0}}\leqslant d_{p}$. Let $s$ be the largest positive integer such that $2sq_{n-n_{0}}\leqslant d_{p}$. Notice that $2(s+1) q_{n-n_0}> d_p$, one has that
\begin{equation*}
    s q_{n-n_0} \geqslant \frac{1}{4} d_p \geqslant \frac{\varepsilon q_n}{400}. 
\end{equation*}

\noindent Case 1: $p\in [\ell q_{n}+\varepsilon q_{n}, \ell q_{n}+k_{j}-\varepsilon q_{n}]$.

If $p\leqslant \ell q_{n}+\frac{k_{j}}{2}$, we construct
\begin{equation*}
    I_{1}=[-2sq_{n-n_{0}}, 2sq_{n-n_{0}}-1],\quad I_{2}=[p-2sq_{n-n_{0}}, p-1].
\end{equation*}

If $p>\ell q_{n}+\frac{k_{j}}{2}$, we construct
\begin{equation*}
    I_{1}=[-2sq_{n-n_{0}}, 2sq_{n-n_{0}}-1],\quad I_{2}=[p+1, p+2sq_{n-n_{0}}].
\end{equation*}

\noindent Case 2: $p\in [\ell q_{n}+k_{j}+\varepsilon q_{n}, (\ell+1)q_{n}-\varepsilon q_{n}]$.

If $p\leqslant (\ell+\frac{1}{2}) q_{n}+\frac{k_{j}}{2}$, we construct
\begin{equation*}
    I_{1}=[-2sq_{n-n_{0}}, 2sq_{n-n_{0}}-1],\quad I_{2}=[p-2sq_{n-n_{0}}, p-1].
\end{equation*}

If $p>(\ell+\frac{1}{2}) q_{n}+\frac{k_{j}}{2}$, we construct
\begin{equation*}
    I_{1}=[-2sq_{n-n_{0}}, 2sq_{n-n_{0}}-1],\quad I_{2}=[p+1, p+2sq_{n-n_{0}}].
\end{equation*}

By the construction of $I_1$ and $I_2$, for any $j_1, j_2 \in I_1 \cup I_2$ with $j_1 \neq j_2$, there exist integers $l_{1},l_{2},l_{1}',l_{2}'$ with $|l_{1}|, |l_{1}'|\leqslant 100 \frac{b_{n+1}}{q_n}+4$ and $1\leqslant l_{2},l_{2}'<q_{n}$ such that
\begin{equation*}
    j_1-j_2=l_1 q_n+l_2,
\end{equation*}
and
\begin{equation*}
    j_1+j_2=k_j+l_1' q_{n}+l_{2}'.
\end{equation*}
Therefore
\begin{equation}\label{reson1}
    \begin{split}
        \|(j_{1}-j_{2})\alpha\|&\geqslant \|l_{2}\alpha\|-|l_{1}| \|q_{n}\alpha\|\\
        &\geqslant \|q_{n-1}\alpha\|-(100 \frac{b_{n+1}}{q_{n}}+4) \|q_{n}\alpha\|\\
        &\geqslant \frac{1}{2q_{n}}- \frac{10^{-5}\eta}{q_{n}}-\frac{4}{q_{n+1}}\\
        &>\frac{1}{4q_{n}}.
    \end{split}
\end{equation}
Since for $n=2j-1$, by \eqref{Ij} we have
\begin{equation*}
    \|2\theta+k_{j}\alpha\|_{\mathbb{R}/\mathbb{Z}} \leqslant \frac{10\eta}{q_{2j}}= \frac{10\eta}{q_{n+1}}.
\end{equation*}
Hence
\begin{equation}\label{reson2}
    \begin{split}
        \|2\theta+(j_{1}+j_{2})\alpha\|&\geqslant \|l_{2}'\alpha\|-\|2\theta+k_{j}\alpha\|_{\mathbb{R}/\mathbb{Z}} - |l_{1}'| \|q_{n}\alpha\|_{\mathbb{R}/\mathbb{Z}}\\
        &\geqslant \|q_{n-1}\alpha\|-\frac{10\eta}{q_{n+1}}-(100\frac{b_{n+1}}{q_{n}} +4)\|q_{n}\alpha\| \\
        &\geqslant \frac{1}{2q_{n}}-\frac{10\eta}{q_{n+1}}-\frac{10^{-5}\eta}{q_{n}}-\frac{4}{q_{n+1}}\\
        &>\frac{1}{4q_{n}}.
    \end{split}
\end{equation}

We now consider the case $n=2j$. For any $p \in [\ell q_n, (\ell+1) q_n ]$ satisfying $|p-\ell q_n| \geqslant \varepsilon q_n$ and $|p-(\ell+1) q_n| \geqslant \varepsilon q_n$, we define
\begin{equation*}
    d_{p}=\frac{1}{100} \min \{\dist(p,q_{n}\mathbb{Z}),10^{-5}\eta q_{n} \}.
\end{equation*}
Let $n_{0}$ be the smallest integer such that $2q_{n-n_{0}}\leqslant d_{p}$. Let $s$ be the largest positive integer such that $2sq_{n-n_{0}}\leqslant d_{p}$. We also have
\begin{equation*}
    s q_{n-n_0} \geqslant \frac{1}{4} d_p \geqslant \frac{\varepsilon q_n}{400}.
\end{equation*}
We construct
\begin{equation*}
    I_{1}=[-sq_{n-n_{0}}, sq_{n-n_{0}}-1],\quad I_{2}=[p-sq_{n-n_{0}}, p+sq_{n-n_{0}}-1]
\end{equation*}
For any $j_1, j_2 \in I_1 \cup I_2$ with $j_1 \neq j_2$, there exist integers $l_{1},l_{2},l_{1}',l_{2}'$ with with $|l_{1}|, |l_{1}'|\leqslant 100 \frac{b_{n+1}}{q_n}+4$ and $1\leqslant l_{2}<q_{n}, 0\leqslant l_{2}'<q_{n}$ such that
\begin{equation*}
    j_1-j_2=l_1 q_n+l_2,
\end{equation*}
and
\begin{equation*}
    j_1+j_2=k_j+l_1' q_{n}+l_{2}'.
\end{equation*}
Therefore,
\begin{equation}\label{reson11}
    \begin{split}
        \|(j_{1}-j_{2})\alpha\|&\geqslant \|l_{2}\alpha\|-|l_{1}| \|q_{n}\alpha\|\\
        &\geqslant \|q_{n-1}\alpha\|-(100 \frac{b_{n+1}}{q_{n}}+4) \|q_{n}\alpha\|\\
        &\geqslant \frac{1}{2q_{n}}- \frac{10^{-5}\eta}{q_{n}}-\frac{4}{q_{n+1}}\\
        &>\frac{1}{4q_{n}}. 
    \end{split}
\end{equation}
If $l_{2}'=0$, then by \eqref{Ij} and $n=2j$,
\begin{equation*}
    \begin{split}
        \|2\theta+(j_{1}+j_{2})\alpha\|&\geqslant \|2\theta+k_{j}\alpha\|_{\mathbb{R}/\mathbb{Z}} - |l_{1}'| \|q_{n}\alpha\|_{\mathbb{R}/\mathbb{Z}}\\
        &\geqslant \frac{\eta}{10q_{n}}-(100\frac{b_{n+1}}{q_{n}} +4)\|q_{n}\alpha\| \\
        &\geqslant \frac{\eta}{10q_{n}}-\frac{10^{-5}\eta}{q_{n}}-\frac{4}{q_{n+1}}\\
        &>\frac{\eta}{20q_{n}}.
    \end{split}
\end{equation*}
If $l_{2}'\neq 0$, then
\begin{equation*}
    \begin{split}
        \|2\theta+(j_{1}+j_{2})\alpha\|&\geqslant \|l_{2}'\alpha\|-\|2\theta+k_{j}\alpha\|_{\mathbb{R}/\mathbb{Z}} - |l_{1}'| \|q_{n}\alpha\|_{\mathbb{R}/\mathbb{Z}}\\
        &\geqslant \|q_{n-1}\alpha\|-\frac{10\eta}{q_{n}}-(100\frac{b_{n+1}}{q_{n}} +4)\|q_{n}\alpha\| \\
        &\geqslant \frac{1}{2q_{n}}-\frac{10\eta}{q_{n}}-\frac{10^{-5}\eta}{q_{n}}-\frac{4}{q_{n+1}}\\
        &>\frac{1}{4q_{n}}.
    \end{split}
\end{equation*}
Thus we always have
\begin{equation}\label{reson12}
    \|2 \theta+(j_1+j_2) \alpha\|_{\mathbb{R} / \mathbb{Z}} \geqslant \frac{\eta}{20 q_n}.
\end{equation}

In all cases, let $\theta_m = \theta + m\alpha$ for $m \in I_1 \cup I_2$. Combine the small divisor estimates \eqref{reson1} with \eqref{reson2}, (or \eqref{reson11} with \eqref{reson12}), we obtain the following Lagrange estimate.

\begin{lemma}\label{lagnonresonant}
    For any $m \in I_1 \cup I_2$, we have $\mathrm{Lag}_{m} \leqslant \varepsilon^{2} q_n$.
\end{lemma}
\begin{proof}
    We only prove the case of $n=2j-1$ since the proof for $n=2j$ is similar. By the construction of $I_{1}$ and $I_{2}$, for sufficiently large $n$,
    \begin{equation*}
        \frac{s}{q_{n}}\ln \frac{q_{n-n_{0}+1}}{s}\leqslant \frac{s}{q_{n-n_{0}+1}}\ln \frac{q_{n-n_{0}+1}}{s}\leqslant \frac{\ln q_{n-n_{0}}}{q_{n-n_{0}}}=o(1),
    \end{equation*}
    where we use $q_{n-n_{0}+1}/s>q_{n-n_{0}}$ and $t^{-1}\ln t$ decays for sufficiently large $t>0$. 
    
    Since $\# (I_{1}\cup I_{2})=6sq_{n-n_{0}}<6q_{n-n_{0}+1}$, we apply \hyperref[item: uniformk]{Theorem~\ref*{uniform}\,\textup{(\ref*{item: uniformk})}} with $6s$ in the place of $s$, and with
    \begin{equation*}
        \sigma_{1}=\sigma_{2}=1,\quad n_{1}=n_{2}=n,\quad \Omega_{1}=\Omega_{2}=0,
    \end{equation*}
    one has
    \begin{equation*}
        \begin{split}
            \Lag_{m}&\leqslant -12s \ln \frac{s}{q_{n-n_{0}+1}} + Cs \ln q_{n-n_{0}} + C \ln q_n\\
            &\leqslant q_{n} \bigg(\frac{12s}{q_{n}}\ln \frac{q_{n-n_{0}+1}}{s} + C\frac{\ln q_{n-n_{0}}}{q_{n-n_{0}}} + C\frac{\ln q_{n}}{q_{n}} \bigg)\\
            &\leqslant \varepsilon^{2} q_{n}.
        \end{split}
    \end{equation*}
\end{proof}

Let $k=\#(I_{1}\cup I_{2})-1$. By \hyperref[lowlag]{Lemma~\ref*{lowlag}}, there exists $m_{0}\in I_{1}\cup I_{2}$ such that
\begin{equation}\label{Pk}
    \bigg|P_{k}\bigg(\theta_{m_{0}}-\frac{k-1}{2}\alpha\bigg)\bigg|\geqslant e^{k L-\varepsilon^{2} q_{n}}.
\end{equation}
Assume $m_{0}\in I_{2}$. Construct $I(p)=[x_{1},x_{2}]$ where in the odd case,
\begin{equation*}
    I(p)=[m_{0}-3sq_{n-n_{0}}+1, m_{0}+3sq_{n-n_{0}}-1],
\end{equation*}
and in the even case,
\begin{equation*}
    I(p)=[m_{0}-2s q_{n-n_{0}}+1, m_{0}+2sq_{n-n_{0}}-1].
\end{equation*}
Then $p\in I(p)$ and $|I(p)|>\frac{\varepsilon q_{n}}{200}$. Moreover,
\begin{equation*}
    \operatorname{dist}(p,\partial I(p))\geqslant sq_{n-n_{0}}-1\geqslant \frac{1}{40} |I(p)|.
\end{equation*}
Moreover, by \eqref{Pk} and \eqref{pregreen}, we have
\begin{equation*}
    |G_{[x_{1},x_{2}]}(p,x_{i})|\leqslant e^{-(L-\varepsilon/2)|p-x_{i}|},\quad \text{for}\ i=1,2.
\end{equation*}
Thus every such $p$ is $(L-\varepsilon/2, |I(p)|)$-regular. Let $[y_{1}, y_{2}]$ be one of $[\ell q_{n},\ell q_{n}+k_{j}]$, $[\ell q_{n}+k_{j}, (\ell+1) q_{n}]$, or $[\ell q_{n}, (\ell+1)q_{n}]$. Set
\begin{equation*}
    \widetilde y_{1}=y_{1}+\frac{9}{10}\varepsilon q_{n},\qquad
    \widetilde y_{2}=y_{2}-\frac{9}{10}\varepsilon q_{n}.
\end{equation*}
Applying \hyperref[block1]{Theorem~\ref*{block1}} to $[\widetilde y_{1},\widetilde y_{2}]$ with
\begin{equation*}
    K=\widetilde y_{2}-\widetilde y_{1},\quad
    \tau=L-\frac{\varepsilon}{2},\quad
    \gamma K=\frac{1}{10}\varepsilon q_{n},\quad
    t=|I(p)|,
\end{equation*}
we have
\begin{equation*}
    [\widetilde y_{1}+\gamma K,\widetilde y_{2}-\gamma K]= [y_{1}+\varepsilon q_{n},y_{2}-\varepsilon q_{n}],
\end{equation*}
and one can check that
\begin{equation*}
    \frac{\gamma K}{20}=\frac{\varepsilon q_n}{200}<t
    \leqslant \frac{1}{2}\min \{|p-\widetilde y_{1}|,|p-\widetilde y_{2}|\}.
\end{equation*}
Thus the assumptions of \hyperref[block1]{Theorem~\ref*{block1}} are satisfied. Since $10\gamma K=\varepsilon q_{n}$,
the endpoint local suprema in \hyperref[block1]{Theorem~\ref*{block1}} are bounded by the corresponding quantities $r_{\ell}$, $r_{\ell+\eta}$, and $r_{\ell+1}$. 

We now exclude $m_{0}\in I_{1}$. Suppose that $m_{0}\in I_{1}$. In the odd case, define 
\begin{equation*}
    [x_{1},x_{2}]=[m_{0}-3sq_{n-n_{0}}+1,\,
m_{0}+3sq_{n-n_{0}}-1],
\end{equation*}
and in the even case, define
\begin{equation*}
    [x_{1},x_{2}]=[m_{0}-2sq_{n-n_{0}}+1,\,
m_{0}+2sq_{n-n_{0}}-1]
\end{equation*}
In both cases, $0\in [x_{1},x_{2}]$.
Then, combining \eqref{Pk} and \eqref{pregreen}, 
\begin{equation*}
    |\phi(0)|\leqslant |\phi(x_{1}-1)| e^{-L|x_{1}|+C\varepsilon^{2} q_{n}} + |\phi(x_{2}+1)| e^{-L|x_{2}|+C\varepsilon^{2} q_{n}}.
\end{equation*}
Since $|x_{i}|\geqslant \frac{\varepsilon q_{n}}{400}-1$ and $\phi$ is a generalized eigenfunction, we have $|\phi(0)|<\frac{1}{2}$ for sufficiently small $\varepsilon$ and large $n$, which contradicts to $\phi(0)=1$.  
\end{proof}

\section{Localization: Resonant sites}
\label{sec:Resonant}
Throughout this section, let $L=\ln |\lambda|$ and assume $L>\max \{\kappa,\mu\}$.
Define
\begin{equation*}
    \beta_{n}\coloneq \frac{\ln q_{n+1}}{q_{n}},\quad n\in\mathbb{N}.
\end{equation*}
By the construction of $\alpha$ via \eqref{an}, we have
\begin{equation*}
    \beta_{2j-1}=\kappa\eta+o(1),\quad \beta_{2j}=\mu+o(1).
\end{equation*}
Moreover, by \hyperref[appkj]{Lemma~\ref*{appkj}}, for $n=2j-1$,
\begin{equation*}
    \beta_n q_n=\kappa k_j+o(k_j).
\end{equation*}
Let $C$ be a large constant depending on $L,\alpha$. We choose $\varepsilon>0$ sufficiently small depending on $L,\eta,\kappa,\mu$. For sufficiently large $n$, the following hold:
\begin{align}
    C\varepsilon q_n + \beta_n q_n - Lk_j 
    &\leqslant -\frac{\eta}{2}(L-\kappa)q_n, 
    && n = 2j-1, \label{absorb1} \\
    C\varepsilon q_n - L(q_n-k_j) + \beta_n q_n 
    &\leqslant -\frac{L}{2}q_n, 
    && n = 2j-1, \label{absorb2} \\
    C\varepsilon q_n - Lq_n + \beta_n q_n 
    &\leqslant -\frac{L-\mu}{2}q_n, 
    && n = 2j. \label{absorb3}
\end{align}
Indeed, \eqref{absorb1} follows from $L>\kappa$ and 
\begin{equation*}
    k_j=\eta q_n+o(q_n),\quad \beta_n=\kappa\eta+o(1).
\end{equation*}
Estimate \eqref{absorb2} follows from $\eta\leqslant 10^{-2}$ and $L>\kappa$.
Finally, Estimate \eqref{absorb3} follows from $L>\mu$ and $\beta_{2j}=\mu+o(1)$.

For simplicity, we suppress the superscripts $\varepsilon,n$ in
$r_\ell^{\varepsilon,n}$ and $r_{\ell+\eta}^{\varepsilon,n}$.

\subsection{Resonance for $n=2j-1$: Estimate of $r_{\ell+\eta}$}

\begin{proposition}\label{reson_odd1}
    Assume $n=2j-1$ and \begin{equation*}
        |\ell|\leqslant 20 \frac{b_{n+1}}{q_{n}}.
    \end{equation*}
    Then, for sufficiently large $n$,
    \begin{equation*}
        r_{\ell+\eta}\leqslant e^{C\varepsilon q_{n}+\beta_{n}q_{n}}(e^{-Lk_j}r_\ell + e^{-Lq_n+Lk_j}r_{\ell+1}).
    \end{equation*}
    Consequently, 
    \begin{equation*}
        r_{\ell+\eta}\leqslant e^{-\gamma_{1} q_{n}}(r_{\ell}+r_{\ell+1})
    \end{equation*}
    for some $\gamma_{1}=\gamma_{1}(L,\kappa,\eta)>0$.
\end{proposition}

\begin{proof}
Take any $p$ with $|p-(\ell q_{n}+k_{j})|\leqslant 10\varepsilon q_{n}$. Let $n_{0}$ be the least positive integer such that
\begin{equation*}
    q_{n-n_{0}}\leqslant \frac{\varepsilon}{2}\bigg(\frac{k_{j}}{4}-4\varepsilon q_{n}\bigg).
\end{equation*}
Let $s$ be the largest positive integer such that $sq_{n-n_{0}}\leqslant \frac{k_{j}}{4}-4\varepsilon q_{n}$. Since $(s+1)q_{n-n_{0}}>\frac{k_{j}}{4}-4\varepsilon q_{n}$, one has
\begin{equation*}
    \frac{k_{j}}{4}-5\varepsilon q_{n}\leqslant sq_{n-n_{0}}\leqslant \frac{k_{j}}{4}-4\varepsilon q_{n}.
\end{equation*}

We construct
\begin{equation*}
    \begin{split}
        &I_{1}=[-2sq_{n-n_{0}}, 2sq_{n-n_{0}}-1],\\
        &I_{2}=[\ell q_{n}+k_{j}-sq_{n-n_{0}}, \ell q_{n}+k_{j}+sq_{n-n_{0}}-1].
    \end{split}
\end{equation*}
Let $\theta_{m}=\theta+m\alpha$ for $m\in I_{1}\cup I_{2}$. The set $\{\theta_m:m\in I_1\cup I_2\}$ consists of
$6sq_{n-n_0}$ points. Let $k=6sq_{n-n_{0}}-1$.

For any $j_{1},j_{2}\in I_{1}\cup I_{2}$ with $j_{1}\neq j_{2}$, there exist $l_{1},l_{2},l_{1}',l_{2}'$ with $|l_{1}|,|l_{1}'|\leqslant 40\frac{b_{n+1}}{q_{n}}+4$ and $1\leqslant l_{2}<q_{n}, 0\leqslant l_{2}'<q_{n}$ such that
\begin{equation*}
    j_{1}-j_{2}=l_{1}q_{n}+l_{2},
\end{equation*}
and
\begin{equation*}
    j_{1}+j_{2}=k_{j}+l_{1}'q_{n}+l_{2}'.
\end{equation*}
Hence
\begin{equation}\label{odd1frequency}
    \begin{split}
        \|(j_{1}-j_{2})\alpha\|&\geqslant \|l_{2}\alpha\|-|l_{1}|\|q_{n}\alpha\|\\
        &\geqslant \frac{1}{2q_{n}}- (40\frac{b_{n+1}}{q_{n}}+4) \frac{1}{q_{n+1}}\\
        &\geqslant \frac{1}{4q_{n}}.
    \end{split}
\end{equation}
If $l_{2}'\neq 0$, then
\begin{equation}\label{odd1phase}
    \begin{split}
        \|2\theta+ (j_{1}+j_{2})\alpha\|&\geqslant \|l_{2}'\alpha\|-\|2\theta+k_{j}\alpha\|- |l_{1}'| \| q_{n}\alpha\|\\
        &\geqslant \frac{1}{2q_{n}}-\frac{10\eta}{q_{n+1}}-(40\frac{b_{n+1}}{q_{n}}+4) \frac{1}{q_{n+1}}\\
        &\geqslant \frac{1}{4q_{n}}.
    \end{split}
\end{equation}
If $l_{2}'=0$, then the small divisor may be of size $1/q_{n+1}$. More precisely, we decompose two subcases. If $l_{1}'=0$, then immediately
\begin{equation*}
    \|2\theta+ (j_{1}+j_{2})\alpha\|= \|2\theta+k_{j}\alpha\|\geqslant \frac{\eta}{10q_{n+1}}.
\end{equation*}
If $l_{1}'\neq 0$, then $\|l_{1}'q_{n}\alpha\|<\frac{1}{2}$ and hence
\begin{equation*}
    \begin{split}
        \|2\theta+ (j_{1}+j_{2})\alpha\|&\geqslant \| l_{1}'q_{n}\alpha\|-\|2\theta+k_{j}\alpha\|\\
        &\geqslant \frac{1}{2q_{n+1}}- \frac{10\eta}{q_{n+1}}\\
        &\geqslant \frac{1}{4q_{n+1}}.
    \end{split}
\end{equation*}
Combine both subcases, for $l_{2}'=0$, we have
\begin{equation}\label{odd1phase'}
    \|2\theta+(j_1+j_2)\alpha\|_{\mathbb{R}/\mathbb{Z}}
    \geqslant e^{-(\beta_n+\varepsilon)q_n}.
\end{equation}
For each fixed $j_1$, the exceptional case
$l_2'=0$ can occur for at most one $j_2$. Therefore, compared with the non-resonant Lagrange estimate, one loses at most one factor
$e^{(\beta_n+\varepsilon)q_n}$. 

Based on \eqref{odd1frequency}, \eqref{odd1phase} and \eqref{odd1phase'}, we have the following Lagrange estimate.
\begin{lemma}\label{lag2}
    For any $m \in I_1 \cup I_2$, we have $\mathrm{Lag}_{m} \leqslant (\beta_{n}+\varepsilon) q_n.$
\end{lemma}
\begin{proof}
By the construction of $I_{1}$ and $I_{2}$, one has
    \begin{equation*}
        q_{n-n_{0}+1}>\frac{\varepsilon}{2}\bigg(\frac{k_{j}}{4}-4\varepsilon q_{n}\bigg),
    \end{equation*}
    which means for sufficiently large $n$,
    \begin{equation*}
        \frac{\ln q_{n}}{q_{n-n_{0}+1}}\leqslant \frac{\ln q_{n}}{\varepsilon k_{j}/9}=o(1),
    \end{equation*}
    and
    \begin{equation*}
        \frac{s}{q_{n}}\ln \frac{q_{n}}{s} \leqslant \frac{\ln q_{n-n_{0}}}{q_{n-n_{0}}}=o(1),
    \end{equation*}
    where we use $q_{n}/s>q_{n-n_{0}}$ and $t^{-1}\ln t$ decays for sufficiently large $t>0$.
    
Since $\# (I_{1}\cup I_{2})=6s q_{n-n_{0}}<\frac{3}{2}k_{j} < q_{n}$, we apply \hyperref[item: uniformn]{Theorem~\ref*{uniform}\,\textup{(\ref*{item: uniformn})}} with $6s$ in place of $s$, and with
    \begin{equation*}
        \sigma_{1}=\sigma_{2}=1,\quad n_{1}=n,\quad n_{2}=n+1,\quad \Omega_{1}=0,\quad \Omega_{2}=1,
    \end{equation*}
    one has
    \begin{equation*}
        \begin{split}
            \Lag_{m} &\leqslant 12s \ln \frac{q_{n}}{s} + Cs \ln q_{n-n_{0}} + C \bigg(1+\frac{sq_{n-n_{0}}}{q_{n-n_{0}+1}} \bigg)\ln q_n+\ln q_{n+1}\\
            &\leqslant q_{n}\bigg(\frac{12s}{q_{n}} \ln \frac{q_{n}}{s}+  C \frac{\ln q_{n-n_{0}}}{q_{n-n_{0}}}+C\frac{\ln q_{n}}{q_{n}}+ C \frac{\ln q_{n}}{q_{n-n_{0}+1}} + \frac{\ln q_{n+1}}{q_{n}}\bigg)\\
            &\leqslant q_{n}(\beta_{n}+\varepsilon).
        \end{split}
    \end{equation*}
\end{proof}

By \hyperref[lowlag]{Lemma~\ref*{lowlag}} and \hyperref[lag2]{Lemma~\ref*{lag2}}, there exists some $m_{0}\in I_{1}\cup I_{2}$ such that
\begin{equation*}
    \bigg|P_{k}\bigg(\theta_{m_{0}}-\frac{k-1}{2}\alpha\bigg)\bigg| \geqslant e^{kL-(\beta_{n}+\varepsilon)q_{n}}.
\end{equation*}
Assume $m_{0}\in I_{2}$. Set
\begin{equation*}
    [x_{1},x_{2}]=[m_{0}-3sq_{n-n_{0}}+1, m_{0}+3sq_{n-n_{0}}-1].
\end{equation*}
Thus
\begin{equation}\label{det1}
    |P_{[x_{1},x_{2}]}(\theta)|\geqslant e^{kL-(\beta_{n}+\varepsilon)q_{n}}.
\end{equation}
By \hyperref[block]{Lemma~\ref*{block}} and \eqref{det1}, we have
\begin{equation}\label{phip}
    |\phi(p)|\leqslant e^{\beta_{n}q_{n}+C\varepsilon q_{n}}\sum_{i=1,2} |\phi(x_{i}')| e^{-L |p-x_{i}|},
\end{equation}
where $x_{1}'=x_{1}-1$ and $x_{2}'=x_{2}+1$.

From the construction of $I_{2}$, we have,
\begin{equation*}
    \begin{split}
        &x_{1}'\in [\ell q_{n}+k_{j}-4sq_{n-n_{0}}, \ell q_{n}+k_{j}-2sq_{n-n_{0}}]\subseteq [\ell q_{n}, \ell q_{n}+k_{j}],\\
        &x_{2}'\in [\ell q_{n}+k_{j}+2sq_{n-n_{0}}, \ell q_{n}+k_{j}+4sq_{n-n_{0}}] \subseteq [\ell q_{n}+k_{j}, (\ell+1)q_{n}].
    \end{split}
\end{equation*}
By  the nonresonant estimates in \hyperref[nonres]{Theorem~\ref*{nonres}}, we have,
\begin{equation}\label{boundary1}
    \begin{split}
        |\phi(x_{1}')|&\leqslant r_{\ell} e^{-(L-\varepsilon)(|x_{1}'-\ell q_{n}|-3\varepsilon q_{n})}\\
        &\quad + r_{\ell+\eta} e^{-(L-\varepsilon) (|\ell q_{n}+k_{j}-x_{1}'|-3\varepsilon q_{n})},
    \end{split}
\end{equation}
and
\begin{equation}\label{boundary2}
    \begin{split}
        |\phi(x_{2}')|&\leqslant r_{\ell+\eta}e^{-(L-\varepsilon) (|x_{2}'-(\ell q_{n}+k_{j})|-3\varepsilon q_{n})}\\
        &\quad +r_{\ell+1} e^{-(L-\varepsilon) (|(\ell+1)q_{n}-x_{2}'|-3\varepsilon q_{n})}.
    \end{split}
\end{equation}
Combine \eqref{phip} with \eqref{boundary1} and \eqref{boundary2}, we have
\begin{equation*}
    \begin{split}
        |\phi(p)|&\leqslant e^{\beta_{n}q_{n}+C\varepsilon q_{n}} (r_{\ell}e^{-L|p-x_{1}|} e^{-L|x_{1}'-\ell q_{n}|}\\
        &\qquad\qquad +r_{\ell+\eta} e^{-L|p-x_{1}|} e^{-L|\ell q_{n}+k_{j}-x_{1}'|}\\
        &\qquad\qquad + r_{\ell+\eta} e^{-L|p-x_{2}|} e^{-L|\ell q_{n}+k_{j}-x_{2}'|} \\
        &\qquad\qquad+ r_{\ell+1} e^{-L|p-x_{2}|} e^{-L|(\ell+1)q_{n}-x_{2}'|} ).
    \end{split}
\end{equation*}
Therefore,
\begin{equation}\label{exclude}
    |\phi(p)|\leqslant e^{\beta_{n}q_{n}+C\varepsilon q_{n}} (r_{\ell}e^{-L k_{j}} + r_{\ell+\eta} e^{-L k_{j}} +r_{\ell+1} e^{-Lq_{n}+Lk_{j}}).
\end{equation}
Take the supremum over $|p-(\ell q_{n}+k_{j})|\leqslant 10\varepsilon q_{n}$, one has
\begin{equation*}
    r_{\ell+\eta}\leqslant e^{\beta_{n}q_{n}+C\varepsilon q_{n}} (r_{\ell}e^{-L k_{j}} + r_{\ell+\eta} e^{-L k_{j}} +r_{\ell+1} e^{-Lq_{n}+Lk_{j}}).
\end{equation*}
Since $L>\kappa$, for sufficiently large $n$,
\begin{equation*}
    e^{\beta_n q_n+C\varepsilon q_n-Lk_j}
    \leqslant
    \frac{1}{2}.
\end{equation*}
Thus
\begin{equation*}
    r_{\ell+\eta}\leqslant e^{\beta_{n}q_{n}+C\varepsilon q_{n}} (r_{\ell}e^{-L k_{j}} + r_{\ell+1} e^{-Lq_{n}+Lk_{j}} ).
\end{equation*}

\begin{figure}[htbp]
\centering
\begin{tikzpicture}[x=1cm,y=1cm,line cap=round,line join=round]

\draw[->,line width=0.6pt] (-0.3,0) -- (10.3,0);

\coordinate (F0) at (0.9,0);
\coordinate (P)  at (5.0,0);
\coordinate (F1) at (9.1,0);

\node[circle,fill=black,inner sep=1.8pt] at (F0) {};
\node[circle,draw=black,fill=white,inner sep=1.7pt,line width=0.65pt] at (P) {};
\node[circle,fill=black,inner sep=1.8pt] at (F1) {};

\node[font=\small,below=5pt] at (F0) {$\ell q_n$};
\node[font=\small,below=5pt] at (P) {$\ell q_n+k_j$};
\node[font=\small,below=5pt] at (F1) {$(\ell+1)q_n$};

\draw[line width=1.7pt] (1.75,0) -- (4.05,0);
\node[font=\small,below=5pt] at (2.90,0) {$x_1'$};

\draw[line width=1.7pt] (5.95,0) -- (8.25,0);
\node[font=\small,below=5pt] at (7.10,0) {$x_2'$};

\end{tikzpicture}
    \caption{Odd case $n=2j-1$: diagram for $r_{\ell+\eta}$.}
    \label{fig:odd1}
\end{figure}
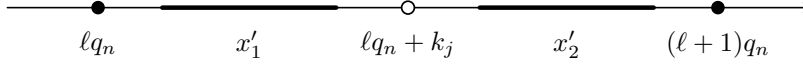

In order to prove the theorem, it suffices to exclude the case $m_{0}\in I_{1}$.
If $m_{0}\in I_{1}$, following the proof of \eqref{exclude} (move $-q_n-k_j$ units in \eqref{exclude}), we have 
\begin{equation}\label{Contradiction}
    |\phi(0)|\leqslant e^{\beta_{n}q_{n}+C\varepsilon q_{n}} (r_{\eta}e^{-L k_{j}} + r_{0} e^{-L k_{j}} +r_{-1+\eta} e^{-Lq_{n}+Lk_{j}}).
\end{equation}
Then  for sufficiently large $n$, by \eqref{absorb1} and \eqref{absorb2}, 
we have that $|\phi(0)|\leqslant \frac{1}{2}$, which contradicts to $\phi(0)=1.$  

Finally, we let 
\begin{equation*}
    \gamma_{1}=\frac{1}{2}\min\bigg\{ \frac{\eta(L-\kappa)}{2}, \frac{L}{2}\bigg\}.
\end{equation*}
Since $n=2j-1$, we deduce from \eqref{absorb1} and \eqref{absorb2} that
\begin{equation*}
    r_{\ell+\eta}\leqslant e^{-\gamma_{1}q_{n}} (r_{\ell}+r_{\ell+1}).
\end{equation*}
This completes the proof.
\end{proof}

\subsection{Resonance for $n=2j-1$: Estimate of $r_{\ell}$}
\begin{proposition}\label{reson_odd2}
    Assume $n=2j-1$ and 
    \begin{equation*}
        0<|\ell|\leqslant 20\frac{b_{n+1}}{q_{n}}.
    \end{equation*} 
    Then, for sufficiently large $n$,
    \begin{equation}\label{rell-odd-short}
        r_{\ell}\leqslant e^{C\varepsilon q_{n}+\beta_{n}q_{n}} (r_{\ell-1+\eta}e^{-L(q_{n}-k_{j})} +r_{\ell+\eta}e^{-Lk_{j}}).
    \end{equation}
    Consequently, 
    \begin{equation*}
        r_{\ell}\leqslant e^{-\gamma_{2}q_{n}}(r_{\ell-1+\eta}+r_{\ell+\eta})
    \end{equation*}
    for some $\gamma_{2}=\gamma_{2}(L,\kappa,\eta)>0$.
\end{proposition}
\begin{proof}
Take any $p$ with $|p-\ell q_{n}|\leqslant 10\varepsilon q_{n}$.  Let $n_{0}$ be the least positive integer such that
\begin{equation*}
    q_{n-n_{0}}\leqslant \frac{\varepsilon}{2}\bigg(\frac{k_{j}}{4}-4\varepsilon q_{n}\bigg).
\end{equation*}
Let $s$ be the largest positive integer such that $sq_{n-n_{0}}\leqslant \frac{k_{j}}{4}-4\varepsilon q_{n}$. Since $(s+1)q_{n-n_{0}}>\frac{k_{j}}{4}-4\varepsilon q_{n}$, one has
\begin{equation*}
    \frac{k_{j}}{4}-5\varepsilon q_{n}\leqslant sq_{n-n_{0}} \leqslant \frac{k_{j}}{4}-4\varepsilon q_{n}.
\end{equation*}

We construct
\begin{equation*}
    \begin{split}
        &I_{1}=[-2sq_{n-n_{0}}, 2sq_{n-n_{0}}-1],\\
        &I_{2}=[\ell q_{n}-sq_{n-n_{0}}, \ell q_{n}+sq_{n-n_{0}}-1].
    \end{split}
\end{equation*}
Let $\theta_{m}=\theta+m\alpha$ for $m\in I_{1}\cup I_{2}$. The set $\{\theta_{m}: m\in I_{1}\cup I_{2}\}$ consists of $6sq_{n-n_{0}}$ points. Let $k=6sq_{n-n_{0}}-1$.

For any $j_{1},j_{2}\in I_{1}\cup I_{2}$ with $j_{1}\neq j_{2}$, there exist $l_{1},l_{2},l_{1}',l_{2}'$ with $|l_{1}|,|l_{1}'|\leqslant 40\frac{b_{n+1}}{q_{n}}+4$ and $0\leqslant l_{2}<q_{n}, 1 \leqslant l_{2}'<q_{n}$ such that
\begin{equation*}
    j_{1}-j_{2}=l_{1}q_{n}+l_{2},
\end{equation*}
and
\begin{equation*}
    j_{1}+j_{2}=k_{j}+l_{1}'q_{n}+l_{2}'.
\end{equation*}
If $l_2=0$, then $l_1\neq 0$. By $|l_1|\|q_n\alpha\|<\frac{1}{2}$, we have
\begin{equation}\label{odd2frequency}
    \|(j_1-j_2)\alpha\|\geqslant
    \|q_n\alpha\|\geqslant \frac{1}{2q_{n+1}}\geqslant e^{-(\beta_n+\varepsilon)q_n}.
\end{equation}
If $l_2\neq0$, then
\begin{equation}\label{odd2frequency'}
    \begin{split}
        \|(j_1-j_2)\alpha\|
        &\geqslant
        \|l_2\alpha\|-|l_1|\|q_n\alpha\|\\
        &\geqslant
        \frac{1}{2q_n}
        -
        (40\frac{b_{n+1}}{q_n}+4)\frac{1}{q_{n+1}}\\
        &\geqslant
        \frac{1}{4q_n}.
    \end{split}
\end{equation}
For each fixed $j_1$, the exceptional case $l_2=0$ can occur for at most one
$j_2\in I_1\cup I_2$.

Since $|l_{2}'|\neq 0$, we have
\begin{equation}\label{odd2phase}
    \begin{split}
        \|2\theta+(j_1+j_2)\alpha\|
        &\geqslant
        \|l_2'\alpha\|
        -
        \|2\theta+k_j\alpha\|
        -
        |l_1'|\|q_n\alpha\|\\
        &\geqslant
        \frac{1}{2q_n}
        -
        \frac{10\eta}{q_{n+1}}
        -
        (40\frac{b_{n+1}}{q_n}+4)\frac{1}{q_{n+1}}\\
        &\geqslant
        \frac{1}{4q_n}.
    \end{split}
\end{equation}

Hence, compared with the non-resonant Lagrange estimate, the only possible extra
small divisor is the frequency one, and it occurs at most once for each fixed $j_1$. 

Based on \eqref{odd2frequency}, \eqref{odd2frequency'} and \eqref{odd2phase}, we have the following Lagrange estimate.

\begin{lemma}\label{lag-ell-short}
    For any $m\in I_{1}\cup I_{2}$, we have $\Lag_{m}\leqslant (\beta_{n}+\varepsilon)q_{n}$.
\end{lemma}
\begin{proof}
By the construction of $I_{1}$ and $I_{2}$, one has
    \begin{equation*}
        q_{n-n_{0}+1}>\frac{\varepsilon}{2}\bigg(\frac{k_{j}}{4}-4\varepsilon q_{n}\bigg),
    \end{equation*}
    which means for sufficiently large $n$,
    \begin{equation*}
        \frac{\ln q_{n}}{q_{n-n_{0}+1}}\leqslant \frac{\ln q_{n}}{\varepsilon k_{j}/9}=o(1),
    \end{equation*}
    and
    \begin{equation*}
        \frac{s}{q_{n}}\ln \frac{q_{n}}{s} \leqslant \frac{\ln q_{n-n_{0}}}{q_{n-n_{0}}}=o(1),
    \end{equation*}
    where we use $q_{n}/s>q_{n-n_{0}}$ and $t^{-1}\ln t$ decays for sufficiently large $t>0$.
    
Since $\# (I_{1}\cup I_{2})=6s q_{n-n_{0}}<\frac{3}{2}k_{j} < q_{n}$, we apply \hyperref[item: uniformn]{Theorem~\ref*{uniform}\,\textup{(\ref*{item: uniformn})}} with $6s$ in place of $s$, and with
    \begin{equation*}
        \sigma_{1}=\sigma_{2}=1,\quad n_{1}=n+1,\quad n_{2}=n,\quad \Omega_{1}=1,\quad \Omega_{2}=0,
    \end{equation*}
    one has
    \begin{equation*}
        \begin{split}
            \Lag_{m} &\leqslant 12s \ln \frac{q_{n}}{s} + Cs \ln q_{n-n_{0}} + C \bigg(1+\frac{sq_{n-n_{0}}}{q_{n-n_{0}+1}} \bigg)\ln q_n+\ln q_{n+1}\\
            &\leqslant q_{n}\bigg(\frac{12s}{q_{n}} \ln \frac{q_{n}}{s}+  C \frac{\ln q_{n-n_{0}}}{q_{n-n_{0}}}+C\frac{\ln q_{n}}{q_{n}}+ C \frac{\ln q_{n}}{q_{n-n_{0}+1}} + \frac{\ln q_{n+1}}{q_{n}}\bigg)\\
            &\leqslant q_{n}(\beta_{n}+\varepsilon).
        \end{split}
    \end{equation*}
\end{proof}

By \hyperref[lowlag]{Lemma~\ref*{lowlag}} and \hyperref[lag-ell-short]{Lemma~\ref*{lag-ell-short}}, there exists
$m_0\in I_1\cup I_2$ such that
\begin{equation*}
    \bigg|
    P_k\bigg(\theta_{m_0}-\frac{k-1}{2}\alpha\bigg)
    \bigg|
    \geqslant
    e^{kL-(\beta_n+\varepsilon)q_n}.
\end{equation*}
Assume that $m_0\in I_2$. Construct
\begin{equation*}
    [x_{1},x_{2}]=[m_{0}-3sq_{n-n_{0}}+1, m_{0}+3sq_{n-n_{0}}-1].
\end{equation*}
Thus
\begin{equation}\label{det2}
    |P_{[x_{1},x_{2}]}(\theta)|\geqslant e^{kL-(\beta_n+\varepsilon)q_n}.
\end{equation}
By \hyperref[block]{Lemma~\ref*{block}} and \eqref{det2}, we have
\begin{equation}\label{green-ell-short}
    |\phi(p)| \leqslant e^{\beta_{n}q_{n}+C\varepsilon q_{n}} \sum_{i=1,2} |\phi(x_{i}')|e^{-L|p-x_{i}|},
\end{equation}
where $x_{1}'=x_{1}-1$ and $x_{2}'=x_{2}+1$.

From the construction of $I_{2}$, we have,
\begin{align*}
    x_{1}' &\in [\ell q_{n}-4sq_{n-n_{0}}, \ell q_{n}-2sq_{n-n_{0}}] \subseteq [(\ell-1)q_{n}+k_{j}, \ell q_{n}], \\
    x_{2}' &\in [\ell q_{n}+2sq_{n-n_{0}}, \ell q_{n}+4sq_{n-n_{0}}] \subseteq [\ell q_{n}, \ell q_{n}+k_{j}].
\end{align*}
By \hyperref[nonres]{Theorem~\ref*{nonres}}, we have
\begin{equation}\label{x1-ell-short}
    \begin{split}
        |\phi(x_1')|
        &\leqslant r_{\ell-1+\eta} e^{-(L-\varepsilon)(|x_1'-((\ell-1)q_n+k_j)|-3\varepsilon q_n)}\\
        &\quad+ r_{\ell} e^{-(L-\varepsilon)(|\ell q_n-x_1'|-3\varepsilon q_n)}.
    \end{split}
\end{equation}
Similarly,
\begin{equation}\label{x2-ell-short}
    \begin{split}
        |\phi(x_2')| &\leqslant r_{\ell} e^{-(L-\varepsilon)(|x_2'-\ell q_n|-3\varepsilon q_n)}\\
        &\quad+ r_{\ell+\eta} e^{-(L-\varepsilon)(|\ell q_n+k_j-x_2'|-3\varepsilon q_n)}.
    \end{split}
\end{equation}
Combining \eqref{green-ell-short} with \eqref{x1-ell-short} and \eqref{x2-ell-short}, we get
\begin{equation*}
    \begin{split}
        |\phi(p)| &\leqslant e^{\beta_nq_n+C\varepsilon q_n}
        ( r_{\ell-1+\eta} e^{-L|p-x_1|} e^{-L|x_1'-((\ell-1)q_n+k_j)|}\\
        &\qquad\qquad + r_{\ell} e^{-L|p-x_1|} e^{-L|\ell q_n-x_1'|}\\
        &\qquad\qquad + r_{\ell} e^{-L|p-x_2|} e^{-L|x_2'-\ell q_n|}\\
        &\qquad\qquad + r_{\ell+\eta} e^{-L|p-x_2|} e^{-L|\ell q_n+k_j-x_2'|}).
    \end{split}
\end{equation*}
Therefore,
\begin{equation*}
    |\phi(p)| \leqslant e^{\beta_nq_n+C\varepsilon q_n}
    (e^{-L(q_n-k_j)}r_{\ell-1+\eta}+ e^{-Lk_j}r_{\ell}+ e^{-Lk_j}r_{\ell+\eta}).
\end{equation*}
Taking the supremum over $|p-\ell q_n|\leqslant 10\varepsilon q_n$, we obtain
\begin{equation}\label{rell-before-absorb-short}
    r_{\ell}\leqslant e^{\beta_n q_n+C\varepsilon q_n} (e^{-L(q_n-k_j)}r_{\ell-1+\eta} + e^{-Lk_j}r_{\ell} + e^{-Lk_j}r_{\ell+\eta}).
\end{equation}
By \eqref{absorb1}, for sufficiently large $n$,
\begin{equation*}
    e^{\beta_nq_n+C\varepsilon q_n-Lk_j}
    \leqslant
    \frac{1}{2}.
\end{equation*}
Thus
\begin{equation*}
    r_{\ell} \leqslant e^{C\varepsilon q_n+\beta_n q_n}
     (e^{-L(q_n-k_j)}r_{\ell-1+\eta}+e^{-Lk_j}r_{\ell+\eta}).
\end{equation*}
This proves \eqref{rell-odd-short}.

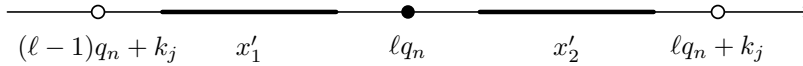
\begin{figure}[htbp]
\centering
\begin{tikzpicture}[x=1cm,y=1cm,line cap=round,line join=round]

\draw[->,line width=0.6pt] (-0.3,0) -- (10.3,0);

\coordinate (P0) at (0.9,0);
\coordinate (F)  at (5.0,0);
\coordinate (P1) at (9.1,0);

\node[circle,draw=black,fill=white,inner sep=1.7pt,line width=0.65pt] at (P0) {};
\node[circle,fill=black,inner sep=1.8pt] at (F) {};
\node[circle,draw=black,fill=white,inner sep=1.7pt,line width=0.65pt] at (P1) {};

\node[font=\small,below=5pt] at (P0) {$(\ell-1)q_n+k_j$};
\node[font=\small,below=5pt] at (F) {$\ell q_n$};
\node[font=\small,below=5pt] at (P1) {$\ell q_n+k_j$};

\draw[line width=1.7pt] (1.75,0) -- (4.05,0);
\node[font=\small,below=5pt] at (2.90,0) {$x_1'$};

\draw[line width=1.7pt] (5.95,0) -- (8.25,0);
\node[font=\small,below=5pt] at (7.10,0) {$x_2'$};

\end{tikzpicture}
    \caption{Odd case $n=2j-1$: diagram for $r_{\ell}$.}
    \label{fig:odd2}
\end{figure}

If $m_{0}\in I_{1}$, following the proof of  \eqref{Contradiction}, we have  $|\phi(0)|\leqslant \frac{1}{2}$, which contradicts to $\phi(0)=1.$

Finally, we also let
\begin{equation*}
    \gamma_2=
    \frac{1}{2}\min\bigg\{\frac{\eta(L-\kappa)}{2},\frac{L}{2}\bigg\}.
\end{equation*}
By \eqref{absorb1} and \eqref{absorb2}, we obtain
\begin{equation*}
    r_{\ell}\leqslant
    e^{-\gamma_2q_n}(
        r_{\ell-1+\eta}+r_{\ell+\eta}).
\end{equation*}
This finishes the proof.
\end{proof}

\subsection{Resonance for $n=2j$}
\begin{proposition}\label{reson-even}
    Assume $n=2j$ and 
    \begin{equation*}
        0<|\ell|\leqslant 20\frac{b_{n+1}}{q_{n}}.
    \end{equation*} 
    Then, for sufficiently large $n$,
    \begin{equation}\label{rell-even}
        r_{\ell}\leqslant e^{-Lq_{n}+\beta_{n} q_{n}+C\varepsilon q_{n}} (r_{\ell-1}+r_{\ell+1}).
    \end{equation}
    Consequently, 
    \begin{equation*}
        r_{\ell}\leqslant e^{-\gamma_{3} q_{n}}(r_{\ell-1}+r_{\ell+1})
    \end{equation*}
    for some $\gamma_{3}=\gamma_{3}(L,\mu)>0$.
\end{proposition}
\begin{proof}
Take any $p$ with $|p-\ell q_{n}|\leqslant 10\varepsilon q_{n}$.
Let $n_{0}$ be the least positive integer such that
\begin{equation*}
    q_{n-n_{0}}\leqslant \frac{\varepsilon}{2}\bigg(\frac{q_{n}}{2}-4\varepsilon q_{n}\bigg).
\end{equation*}
Let $s$ be the largest positive integer such that $sq_{n-n_{0}}\leqslant \frac{q_{n}}{2}-4\varepsilon q_{n}$. Since $(s+1)q_{n-n_{0}}>\frac{q_{n}}{2}-4\varepsilon q_{n}$, one has
\begin{equation*}
    \frac{q_{n}}{2}-5\varepsilon q_{n}\leqslant s q_{n-n_{0}}\leqslant \frac{q_{n}}{2}-4\varepsilon q_{n}.
\end{equation*}

We construct
\begin{equation*}
    \begin{split}
        &I_{1}=[-sq_{n-n_{0}},sq_{n-n_{0}}-1],\\
        &I_{2}=[\ell q_{n}-sq_{n-n_{0}},\ell q_{n}+sq_{n-n_{0}}-1].
    \end{split}
\end{equation*}

Let $\theta_{m}=\theta+m\alpha$ for $m\in I_{1}\cup I_{2}$. The set $\{\theta_{m}:m\in I_{1}\cup I_{2}\}$ consists of $4sq_{n-n_{0}}$ points. Let $k=4sq_{n-n_{0}}-1$.

For any $j_{1},j_{2}\in I_{1}\cup I_{2}$ with $j_{1}\neq j_{2}$, there exist $l_{1},l_{2},l_{1}',l_{2}'$ with $|l_{1}|,|l_{1}'|\leqslant 40\frac{b_{n+1}}{q_{n}}+4$ and $0\leqslant l_{2},l_{2}'<q_{n}$ such that
\begin{equation*}
    j_{1}-j_{2}=l_{1}q_{n}+l_{2},
\end{equation*}
and
\begin{equation*}
    j_{1}+j_{2}=k_{j}+l_{1}'q_{n}+l_{2}'.
\end{equation*}
If $l_{2}=0$, then $\|l_{1}q_{n}\alpha\|<\frac{1}{2}$ and thus
\begin{equation}\label{evenfrequency}
    \|(j_{1}-j_{2})\alpha\| \geqslant \|q_{n}\alpha\|\geqslant \frac{1}{2q_{n+1}}\geqslant e^{-(\beta_{n}+\varepsilon)q_{n}}.
\end{equation}
If $l_{2}\neq 0$, then
\begin{equation}\label{evenfrequency'}
    \begin{split}
        \|(j_{1}-j_{2})\alpha\|&\geqslant \|l_{2}\alpha\|-|l_{1}|\|q_{n}\alpha\|\\
        &\geqslant \frac{1}{2q_{n}}- (40\frac{b_{n+1}}{q_{n}}+4) \frac{1}{q_{n+1}}\\
        &\geqslant \frac{1}{4q_{n}}.
    \end{split}
\end{equation}
For each fixed $j_1$, the exceptional case $l_2=0$ can occur for at most one $j_2$.

If $l_2'=0$, then by \eqref{Ij} and $n=2j$,
\begin{equation}\label{evenphase}
    \begin{split}
        \|2\theta+(j_1+j_2)\alpha\|
    &\geqslant
    \|2\theta+k_j\alpha\|
    -
    |l_1'|\|q_n\alpha\|  \\
    &\geqslant
    \frac{\eta}{10q_n}
    -
    (40\frac{b_{n+1}}{q_n}+4)\frac{1}{q_{n+1}} \\
    &\geqslant
    \frac{\eta}{20q_n}.
    \end{split}
\end{equation}
If $l_2'\neq0$, then
\begin{equation}\label{evenphase'}
    \begin{split}
        \|2\theta+(j_1+j_2)\alpha\|
    &\geqslant
    \|l_2'\alpha\|
    -
    \|2\theta+k_j\alpha\|
    -
    |l_1'|\|q_n\alpha\|  \\
    &\geqslant
    \frac{1}{2q_n}
    -
    \frac{10\eta}{q_n}
    -
    (40\frac{b_{n+1}}{q_n}+4)\frac{1}{q_{n+1}}  \\
    &\geqslant
    \frac{1}{4q_n}.
    \end{split}
\end{equation}
Since $\eta\leqslant 10^{-2}$, in both cases we have
\begin{equation*}
    \|2\theta+(j_1+j_2)\alpha\|\geqslant
    \frac{\eta}{20q_n}.
\end{equation*}
Hence, the only possible extra small divisor is the frequency one, and it occurs at most once.

Based on \eqref{evenfrequency}, \eqref{evenfrequency'}, \eqref{evenphase}, and \eqref{evenphase'}, we have the following Lagrange estimate.

\begin{lemma}\label{lageven}
    For any $m\in I_{1}\cup I_{2}$, we have $\Lag_{m}\leqslant (\beta_{n}+\varepsilon) q_{n}$.
\end{lemma}
\begin{proof}
    By the construction of $I_{1}$ and $I_{2}$, one has
    \begin{equation*}
        q_{n-n_{0}+1}>\frac{\varepsilon}{2}\bigg(\frac{q_{n}}{2}-4\varepsilon q_{n}\bigg),
    \end{equation*}
    which means for sufficiently large $n$,
    \begin{equation*}
        \frac{\ln q_{n}}{q_{n-n_{0}+1}}\leqslant \frac{\ln q_{n}}{\varepsilon q_{n}/5}=o(1),
    \end{equation*}
    and 
    \begin{equation*}
        \frac{s}{q_{n}}\ln \frac{q_{n}}{s}\leqslant \frac{\ln q_{n-n_{0}}}{q_{n-n_{0}}}=o(1),
    \end{equation*}
    where we use $q_{n}/s>q_{n-n_{0}}$ and $t^{-1}\ln t$ decays for sufficiently large $t>0$.

    Since $\# (I_{1}\cup I_{2})=4s q_{n-n_{0}}< 2q_{n}$, we apply \hyperref[item: uniformn]{Theorem~\ref*{uniform}\,\textup{(\ref*{item: uniformn})}} with $4s$ in place of $s$, and with
    \begin{equation*}
        \sigma_{1}=\sigma_{2}=1, \quad n_{1}=n+1, \quad n_{2}=n, \quad \Omega_{1}=1, \quad \Omega_{2}=0,
    \end{equation*}
    one has
    \begin{equation*}
        \begin{split}
            \Lag_{m} &\leqslant 8s \ln \frac{q_{n}}{s} + Cs \ln q_{n-n_{0}} + C \bigg(1+\frac{sq_{n-n_{0}}}{q_{n-n_{0}+1}} \bigg)\ln q_n+\ln q_{n+1}\\
            &\leqslant q_{n}\bigg(\frac{8s}{q_{n}} \ln \frac{q_{n}}{s}+  C \frac{\ln q_{n-n_{0}}}{q_{n-n_{0}}}+C\frac{\ln q_{n}}{q_{n}}+ C \frac{\ln q_{n}}{q_{n-n_{0}+1}} + \frac{\ln q_{n+1}}{q_{n}}\bigg)\\
            &\leqslant q_{n}(\beta_{n}+\varepsilon).
        \end{split}
    \end{equation*}
\end{proof}

By \hyperref[lowlag]{Lemma~\ref*{lowlag}} and \hyperref[lageven]{Lemma~\ref*{lageven}}, there exists $m_0\in I_1\cup I_2$ such that
\begin{equation*}
    \bigg|P_{k}\bigg(\theta_{m_{0}}-\frac{k-1}{2}\alpha\bigg)\bigg|\geqslant e^{kL-(\beta_{n}+\varepsilon) q_{n}}.
\end{equation*}
Assume that $m_0\in I_2$. Construct
\begin{equation*}
    [x_{1},x_{2}] = [m_{0}-2sq_{n-n_{0}}+1, m_{0}+2sq_{n-n_{0}}-1].
\end{equation*}
Thus
\begin{equation}\label{det3}
    |P_{[x_{1},x_{2}]}|\geqslant e^{kL-(\beta_{n}+\varepsilon) q_{n}}
\end{equation}
By \hyperref[block]{Lemma~\ref*{block}} and \eqref{det3}, it follows that
\begin{equation*}
    |\phi(p)| \leqslant e^{\beta_{n} q_{n}+C\varepsilon q_{n}} \sum_{i=1,2} |\phi(x_{i}')|e^{-L|p-x_{i}|},
\end{equation*}
where $x_{1}'=x_{1}-1$ and $x_{2}'=x_{2}+1$.

To proceed, we divide the analysis into two cases.

\noindent Case 1: $m_{0} \leqslant \ell q_{n}$. From the construction of $I_{2}$, we have,
\begin{align*}
    x_{1}' &\in [\ell q_{n}-3sq_{n-n_{0}}, \ell q_{n}-2sq_{n-n_{0}}] \subseteq [(\ell-2)q_{n}, (\ell-1)q_{n}+10\varepsilon q_{n}], \\
    x_{2}' &\in [\ell q_{n}+sq_{n-n_{0}}, \ell q_{n}+2sq_{n-n_{0}}] \subseteq [\ell q_{n}, (\ell+1)q_{n}-8\varepsilon q_{n}].
\end{align*}
If $x_{1}'$ (or $x_{2}'$) is in the $10\varepsilon q_n$ neighborhood of $(\ell-1)q_n$ (or $(\ell+1)q_{n}$), we bound it directly by  $r_{\ell-1}$ (or $r_{\ell+1}$). Otherwise, we apply \hyperref[nonres]{Theorem~\ref*{nonres}}. In both cases, 
\begin{equation*}
    \begin{split}
        |\phi(x_{1}')| &\leqslant r_{\ell-2}e^{-(L-\varepsilon)(|x_{1}'-(\ell-2)q_{n}|-10\varepsilon q_{n})} \\
        &\quad + r_{\ell-1} e^{-(L-\varepsilon)(|(\ell-1)q_{n}-x_{1}'|-10\varepsilon q_{n})}, \\
    \end{split}
\end{equation*}
and
\begin{equation*}
    \begin{split}
        |\phi(x_{2}')| &\leqslant r_{\ell}e^{-(L-\varepsilon)(|x_{2}'-\ell q_{n}|-10\varepsilon q_{n})}\\
        &\quad+ r_{\ell+1}e^{-(L-\varepsilon)(|(\ell+1)q_{n}-x_{2}'|-10\varepsilon q_{n})}.
    \end{split}
\end{equation*}
Consequently, we have
\begin{equation*}
    \begin{split}
        |\phi(p)|& \leqslant e^{\beta_{n} q_{n}+C\varepsilon q_{n}}(r_{\ell-2}e^{-L|p-x_{1}|}e^{-L|x_{1}'-(\ell-2)q_{n}|}\\
        &\qquad\qquad +r_{\ell-1}e^{-L|p-x_{1}|} e^{-L|(\ell-1)q_{n}-x_{1}'|}\\
        &\qquad\qquad + r_{\ell}e^{-L|p-x_{2}|}e^{-L|x_{2}'-\ell q_{n}|}\\
        &\qquad\qquad+r_{\ell+1}e^{-L|p-x_{2}|}e^{-L|(\ell +1)q_{n}-x_{2}'|} ).
    \end{split}
\end{equation*}
Thus
\begin{equation}\label{case1-even}
\begin{split}
    |\phi(p)|
    &\leqslant
    e^{C\varepsilon q_n-2Lq_n+\beta_nq_n}r_{\ell-2}
    +
    e^{C\varepsilon q_n-Lq_n+\beta_nq_n}r_{\ell-1} \\
    &\quad+
    e^{C\varepsilon q_n-Lq_n+\beta_nq_n}r_{\ell}
    +
    e^{C\varepsilon q_n-Lq_n+\beta_nq_n}r_{\ell+1}.
\end{split}
\end{equation}

\noindent Case 2: $m_{0} > \ell q_{n}$. Then
\begin{align*}
    x_{1}' &\in [\ell q_{n}-2sq_{n-n_{0}}, \ell q_{n}-sq_{n-n_{0}}] \subseteq [(\ell-1)q_{n}+8\varepsilon q_{n},\ell q_{n}], \\
    x_{2}' &\in [\ell q_{n}+2sq_{n-n_{0}}, \ell q_{n}+3sq_{n-n_{0}}] \subseteq [(\ell+1)q_{n}-10\varepsilon q_{n}, (\ell+2)q_{n}].
\end{align*}
As in Case 1, when $x_{1}'$ (or $x_{2}'$) is in the  $10\varepsilon q_n$ neighborhood of $(\ell-1)q_{n}$ (or $(\ell+1)q_n$), we bound it directly by $r_{\ell-1}$ (or $r_{\ell+1}$).  Otherwise we apply \hyperref[nonres]{Theorem~\ref*{nonres}}. In both cases,
\begin{equation*}
    \begin{split}
        |\phi(x_{1}')| &\leqslant 
    r_{\ell-1}
    e^{-(L-\varepsilon)(|x_{1}'-(\ell-1)q_{n}|-10\varepsilon q_{n})} \\
    &\quad+r_{\ell}
    e^{-(L-\varepsilon)(|\ell q_{n}-x_{1}'|-10\varepsilon q_{n})},
    \end{split}
\end{equation*}
and
\begin{equation*}
    \begin{split}
        |\phi(x_{2}')| &\leqslant 
    r_{\ell+1}
    e^{-(L-\varepsilon)(|x_{2}'-(\ell+1)q_{n}|-10\varepsilon q_{n})} \\
    &\quad+r_{\ell+2}
    e^{-(L-\varepsilon)(|(\ell+2)q_{n}-x_{2}'|-10\varepsilon q_{n})}.
    \end{split}
\end{equation*}
This implies
\begin{equation*}
    \begin{split}
        |\phi(p)|&\leqslant e^{\beta_{n} q_{n}+C\varepsilon q_{n}} (r_{\ell-1}e^{-L|p-x_{1}|} e^{-L|x_{1}'-(\ell-1)q_{n}|} \\
        &\qquad\qquad+r_{\ell} e^{-L|p-x_{1}|} e^{-L|\ell q_{n}-x_{1}'|}\\
        &\qquad\qquad + r_{\ell+1} e^{-L|p-x_{2}|} e^{-L|x_{2}'-(\ell+1)q_{n}|} \\
        &\qquad\qquad+ r_{\ell+2} e^{-L|p-x_{2}|} e^{-L|(\ell+2)q_{n}-x_{2}'|} ).
    \end{split}
\end{equation*}
Thus
\begin{equation}\label{case2-even}
\begin{split}
    |\phi(p)|
    &\leqslant
    e^{C\varepsilon q_n-Lq_n+\beta_nq_n}r_{\ell-1}
    +
    e^{C\varepsilon q_n-Lq_n+\beta_nq_n}r_{\ell} \\
    &\quad+
    e^{C\varepsilon q_n-Lq_n+\beta_nq_n}r_{\ell+1}
    +
    e^{C\varepsilon q_n-2Lq_n+\beta_nq_n}r_{\ell+2}.
\end{split}
\end{equation}

Combining \eqref{case1-even} and \eqref{case2-even}, we obtain
\begin{equation*}
    \begin{split}
        r_{\ell}&\leqslant e^{-Lq_{n}+\beta_{n}q_{n}+C\varepsilon q_{n}} (r_{\ell-1}+r_{\ell}+r_{\ell+1})\\
        &\quad +e^{-2Lq_{n}+\beta_{n}q_{n}+C\varepsilon q_{n}} (r_{\ell-2}+r_{\ell+2}).
    \end{split}
\end{equation*}
By \eqref{absorb3}, for sufficiently large $n$,
\begin{equation*}
    e^{C\varepsilon q_n-Lq_n+\beta_nq_n}
    \leqslant
    \frac{1}{2}.
\end{equation*}
Thus
\begin{equation}\label{after-absorb-even}
\begin{split}
    r_{\ell}
    &\leqslant
    e^{C\varepsilon q_n-Lq_n+\beta_nq_n}
    (r_{\ell-1}+r_{\ell+1})  \\
    &\quad+
    e^{C\varepsilon q_n-2Lq_n+\beta_nq_n}
    (r_{\ell-2}+r_{\ell+2}).
\end{split}
\end{equation}
By the transfer matrix estimate \eqref{transfermatrix},
\begin{equation*}
    r_{\ell+2}
    \leqslant
    e^{C\varepsilon q_n}e^{Lq_n}r_{\ell+1},
    \quad
    r_{\ell-2}
    \leqslant
    e^{C\varepsilon q_n}e^{Lq_n}r_{\ell-1}.
\end{equation*}
Substituting this into \eqref{after-absorb-even}, we obtain
\begin{equation*}
    r_{\ell}
    \leqslant
    e^{C\varepsilon q_n-Lq_n+\beta_nq_n}
    (r_{\ell-1}+r_{\ell+1}).
\end{equation*}
This proves \eqref{rell-even}.

\begin{figure}[htbp]
\centering
\begin{tikzpicture}[x=1cm,y=1cm,line cap=round,line join=round]

\draw[->,line width=0.6pt] (-0.3,0) -- (10.3,0);

\coordinate (F0) at (0.9,0);
\coordinate (F)  at (5.0,0);
\coordinate (F1) at (9.1,0);

\node[circle,fill=black,inner sep=1.8pt] at (F0) {};
\node[circle,fill=black,inner sep=1.8pt] at (F) {};
\node[circle,fill=black,inner sep=1.8pt] at (F1) {};

\node[font=\small,below=5pt] at (F0) {$(\ell-1)q_n$};
\node[font=\small,below=5pt] at (F) {$\ell q_n$};
\node[font=\small,below=5pt] at (F1) {$(\ell+1)q_n$};

\draw[line width=1.7pt] (1.75,0) -- (4.05,0);
\node[font=\small,below=5pt] at (2.90,0) {$x_1'$};

\draw[line width=1.7pt] (5.95,0) -- (8.25,0);
\node[font=\small,below=5pt] at (7.10,0) {$x_2'$};

\end{tikzpicture}
    \caption{Even case $n=2j$: diagram for $r_{\ell}$.}
    \label{fig:even}
\end{figure}
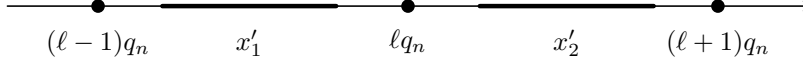

If $m_{0}\in I_{1}$, following the proof of  \eqref{Contradiction}, we have  that $|\phi(0)|\leqslant \frac{1}{2}$, which contradicts to $\phi(0)=1.$

Finally, we let 
\begin{equation*}
    \gamma_{3}=\frac{L-\mu}{4},
\end{equation*}
By \eqref{absorb3} and $n=2j$,  we have
\begin{equation*}
    r_{\ell}\leqslant e^{-\gamma_3 q_n}(r_{\ell-1}+r_{\ell+1}).
\end{equation*}
This finishes the proof.
\end{proof}

\section{Proof of Theorem~\ref*{mainthm}}

We only prove \hyperref[item:AL]{Theorem~\ref*{mainthm}\,\textup{(\ref*{item:AL})}}, since \hyperref[item:SC]{Theorem~\ref*{mainthm}\,\textup{(\ref*{item:SC})}} is well-known. Let
$L = \ln|\lambda| > \max\{\kappa,\mu\}$, and let $\phi$ be a generalized eigenfunction
satisfying \eqref{eigen1} and \eqref{normalization}. Choose
\begin{equation*}
    0 < \gamma < \frac{1}{100}\min\{L, \gamma_1, \gamma_2, \gamma_3\}.
\end{equation*}
Let $b_n = 10^{-7}\eta q_n$, where $\eta$ satisfies \eqref{eta}. For sufficiently large $n$, Propositions~\ref{reson_odd1}, \ref{reson_odd2}, and
\ref{reson-even} imply that for $n = 2j$,
\begin{equation}\label{main-even-rec}
    r_{\ell} \leqslant e^{-8\gamma q_n}\max\{r_{\ell-1}, r_{\ell+1}\},
    \qquad 0 < |\ell| \leqslant 20b_{n+1}/q_n,
\end{equation}
and for $n = 2j-1$,
\begin{equation}\label{main-odd-rec1}
    r_{\ell+\eta} \leqslant e^{-8\gamma q_n}\max\{r_\ell, r_{\ell+1}\},
    \qquad |\ell| \leqslant 20b_{n+1}/q_n,
\end{equation}
while
\begin{equation}\label{main-odd-rec2}
    r_{\ell} \leqslant e^{-8\gamma q_n}\max\{r_{\ell-1+\eta}, r_{\ell+\eta}\},
    \qquad 0 < |\ell| \leqslant 20b_{n+1}/q_n.
\end{equation}
Combining \eqref{main-odd-rec1} and
\eqref{main-odd-rec2}, and absorbing the $r_\ell$ term, we again obtain
\begin{equation}\label{main-int-rec}
    r_{\ell} \leqslant e^{-8\gamma q_n}\max\{r_{\ell-1}, r_{\ell+1}\},
    \qquad 0 < |\ell| \leqslant 20b_{n+1}/q_n.
\end{equation}

Let $M_n = \lfloor 20\frac{b_{n+1}}{q_n}\rfloor - 1$.
By iterating \eqref{main-even-rec} and
\eqref{main-int-rec}, for $0 \leqslant \ell \leqslant M_n$, we have
\begin{equation*}
    r_\ell \leqslant
    \max\{e^{-4\gamma\ell q_n}r_0,\,
    e^{-4\gamma(M_n-\ell)q_n}r_{M_n}\}.
\end{equation*}
Since $\phi$ is a generalized eigenfunction, $r_0 \leqslant Cq_n$ and
$r_{M_n} \leqslant Cq_{n+1}$. Hence, for sufficiently large $n$ and
$1 \leqslant \ell \leqslant 10b_{n+1}/q_n$,
\begin{equation}\label{main-r-decay}
    r_\ell \leqslant e^{-2\gamma\ell q_n}.
\end{equation}
A similar argument yields
\begin{equation}\label{main-r-decay-neg}
    r_\ell \leqslant e^{-2\gamma|\ell| q_n},
    \qquad
    1 \leqslant |\ell| \leqslant 10b_{n+1}/q_n.
\end{equation}
If $n = 2j-1$, then \eqref{main-odd-rec1}, \eqref{main-r-decay}, and
\eqref{main-r-decay-neg} imply
\begin{equation}\label{main-phase-decay}
    r_{\ell+\eta} \leqslant e^{-2\gamma|\ell+\eta|q_n},
    \qquad
    |\ell| \leqslant 10b_{n+1}/q_n.
\end{equation}

Now, take $|k|$ to be sufficiently large and choose $n$ such that $8b_n \leqslant |k| < 8b_{n+1}$. Thus $|k|> 10\varepsilon q_{n}$ for sufficiently small $\varepsilon>0$.
We may assume $k > 0$, as the case $k < 0$ is analogous. If $\operatorname{dist}(k, q_n\mathbb{Z}) < 10\varepsilon q_n$, then \eqref{main-r-decay} gives
\begin{equation}\label{ellres}
    |\phi(k)| \leqslant e^{-(\gamma-\varepsilon)k}.
\end{equation}
If $n = 2j-1$ and $\operatorname{dist}(k, q_n\mathbb{Z} + k_j) < 10\varepsilon q_n$, then by \eqref{ks} and \eqref{main-phase-decay},
\begin{equation}\label{etares}
    |\phi(k)| \leqslant e^{-(\gamma-\varepsilon)k}.
\end{equation}

It remains to consider non-resonant $k$. If $n = 2j$, we choose $\ell$ such that
$k \in [\ell q_n, (\ell+1)q_n]$. By Theorem~\ref{nonres},
\begin{equation*}
    |\phi(k)| \leqslant
        r_{\ell}e^{-(L-\varepsilon)(|k-\ell q_n|-3\varepsilon q_n)} + r_{\ell+1}e^{-(L-\varepsilon)(|(\ell+1)q_n-k|-3\varepsilon q_n)}.
\end{equation*}
Using \eqref{main-r-decay}, \eqref{main-r-decay-neg}, and the assumption $L > 100\gamma$, we obtain
\begin{equation}\label{evennonres}
    |\phi(k)| \leqslant e^{-(\gamma-\varepsilon)k}.
\end{equation}
If $n = 2j-1$, then $k$ belongs to one of the two intervals $[\ell q_n, \ell q_n+k_j]$ or $[\ell q_n+k_j, (\ell+1)q_n]$.
Applying the corresponding estimate in Theorem~\ref{nonres}, and subsequently using
\eqref{main-r-decay}, \eqref{main-r-decay-neg}, \eqref{main-phase-decay}, and
$k_j = \eta q_n + o(q_n)$, we again obtain
\begin{equation}\label{oddnonres}
    |\phi(k)| \leqslant e^{-(\gamma-\varepsilon)k}.
\end{equation}
Combining \eqref{ellres} and \eqref{etares} with \eqref{evennonres} and \eqref{oddnonres} completes the proof.
\qed

\appendix
\section{Lagrange interpolation}
\begin{theorem}\label{uniform}
    Let $k=n-n_{0}$. Let $I_{1}, I_{2}\subseteq \mathbb{Z}$ disjoint, and set $I=I_{1}\cup I_{2}$. Assume $I_{1}$ and $I_{2}$ are each a union of consecutive intervals of length $q_{k}$. Fix $n_{1},n_{2}\geqslant n$. Assume
    \begin{equation*}
        \min_{i,j\in I, i\neq j}\|(i-j)\alpha\|_{\mathbb{R}/\mathbb{Z}} \geqslant \frac{1}{C_{0}q_{n_{1}}^{\sigma_{1}}},
    \end{equation*}
    and 
    \begin{equation*}
        \min_{i,j\in I}\|2\theta+(i+j)\alpha\|_{\mathbb{R}/\mathbb{Z}} \geqslant \frac{1}{C_{0}q_{n_{2}}^{\sigma_{2}}}.
    \end{equation*}
    Moreover, assume that
    \begin{equation*}
            \sup_{i\in I}\# \bigg\{j\in I, j\neq i: \|(i-j)\alpha\|< \frac{1}{C_{0}'q_{n}^{\sigma_{1}}} \bigg\}\leqslant \Omega_{1},
        \end{equation*}
        and
        \begin{equation*}
            \sup_{i\in I}\# \bigg\{j\in I: \|2\theta+(i+j)\alpha\|< \frac{1}{C_{0}'q_{n}^{\sigma_{2}}} \bigg\}\leqslant \Omega_{2},
        \end{equation*}
    where $C_{0}'\geqslant C_{0}$ is fixed. Then the following hold.

    \begin{enumerate}
        \item \label{item: uniformk} If $\# I=sq_{k}< C_{0}'' q_{k+1}$, then for every $m\in I$,
\begin{equation*}
    \begin{split}
        \mathrm{Lag}_m &\leqslant -2s \ln \frac{s}{q_{k+1}} + Cs \ln q_{k} + C(\sigma_{1}+\sigma_{2}) \ln q_n\\
        &\qquad+\Omega_{1}\sigma_{1}\ln q_{n_{1}}+\Omega_{2}\sigma_{2}\ln q_{n_{2}}.
    \end{split}
\end{equation*}

        \item \label{item: uniformn} If $\# I=sq_{k}< C_{0}'' q_{n}$, then for every $m\in I$,
\begin{equation*}
    \begin{split}
        \Lag_{m}\leqslant &-2s \ln \frac{s}{q_{n}}+Cs \ln q_{k} + C(\sigma_{2}+\sigma_{1})(1+\frac{sq_{k}}{q_{k+1}})\ln q_{n}\\
        &\qquad+\Omega_{1}\sigma_{1}\ln{q_{n_{1}}}+\Omega_{2}\sigma_{2}\ln{q_{n_{2}}}.
    \end{split}
\end{equation*}
    \end{enumerate}
\end{theorem}

We need the following estimate.
\begin{lemma}[\cite{MR2521117}]\label{av}
For $\alpha \in \mathbb{R} \setminus \mathbb{Q}$, $\theta \in \mathbb{R}$ and $0 \leqslant j_{0} \leqslant q_{n}-1$  be such that
\begin{equation*}
    |\cos \pi (\theta+j_{0} \alpha )|=\inf _{0 \leqslant j \leqslant q_{n}-1}|\cos \pi(\theta+j \alpha)|,
\end{equation*}
then for some absolute constant $C >0$
\begin{equation}\label{avila-tri}
    \bigg|\sum_{j=0, j \neq j_{0}}^{q_{n}-1} \ln |\cos \pi(\theta+j\alpha)|+(q_{n}-1) \ln 2\bigg| \leqslant C \ln q_{n}.
\end{equation}
\end{lemma}

\begin{proof}[Proof of Theorem \ref{uniform}]
We provide the proof of \hyperref[item: uniformk]{Theorem~\ref*{uniform}\,\textup{(\ref*{item: uniformk})}} in detail then sketch the proof for \hyperref[item: uniformn]{Theorem~\ref*{uniform}\,\textup{(\ref*{item: uniformn})}}.
Fix $m \in I_1 \cup I_2$. For any $a \in \mathbb{T}$, we have
\begin{equation}\label{logm1}
    \begin{split}
        &\ln \prod_{\ell \neq m} \frac{|\cos 2\pi a - \cos 2\pi \theta_\ell|}{|\cos 2\pi \theta_m - \cos 2\pi \theta_\ell|} \\
        &= \sum_{\ell \neq m} \ln |\cos 2\pi a - \cos 2\pi \theta_\ell| - \sum_{\ell \neq m} \ln |\cos 2\pi \theta_m - \cos 2\pi \theta_\ell|.
    \end{split}
\end{equation}
We first estimate $\sum_{\ell \neq m} \ln |\cos 2\pi a - \cos 2\pi \theta_\ell|$. Recall $\# I = s q_{k}$. Rewrite
\begin{equation}\label{logm2}
    \begin{split}
        \sum_{\ell \neq m} &\ln |\cos 2\pi a - \cos 2\pi \theta_\ell| \\
        &= \sum_{\ell \neq m} \left( \ln |\sin \pi (a + \theta_\ell)| + \ln |\sin \pi (a - \theta_\ell)| + \ln 2 \right) \\
        &\eqcolon\Sigma_+ + \Sigma_- + (s q_{k} - 1) \ln 2,
    \end{split}
\end{equation}
where $\Sigma_{\pm}$ consist of $s$ terms of the form \eqref{avila-tri}, supplemented by $s$ terms of the form
\begin{equation*}
    \ln \inf_{0 \leqslant \ell \leqslant q_{k}-1} |\sin \pi (x + \ell \alpha)|,
\end{equation*}
and excluding the term $\ln |\sin \pi (a \pm \theta_m)|$. Applying \hyperref[av]{Lemma~\ref*{av}} to each of the $s$ components in $\Sigma_{\pm}$, we obtain
\begin{equation*}
    \Sigma_{\pm} \leqslant -s q_{k} \ln 2 + Cs \ln q_{k}.
\end{equation*}
Consequently, 
\begin{equation}\label{num01}
    \sum_{\ell \neq m} \ln |\cos 2\pi a - \cos 2\pi \theta_\ell| \leqslant -s q_{k} \ln 2 + Cs \ln q_{k}.
\end{equation}

Next, we estimate $\sum_{\ell \neq m} \ln |\cos 2\pi \theta_m - \cos 2\pi \theta_\ell|$. Write
\begin{equation*}
    \begin{split}
        \sum_{\ell \neq m} \ln |\cos 2\pi \theta_m - \cos 2\pi \theta_\ell| &= \Sigma_+ + \Sigma_- + (s q_{k} - 1) \ln 2,
    \end{split}
\end{equation*}
where
\begin{equation*}
    \Sigma_+ = \sum_{\ell \neq m} \ln |\sin \pi (2\theta + (m+\ell)\alpha)|, 
\end{equation*}
and
\begin{equation*}
    \Sigma_- = \sum_{\ell \neq m} \ln |\sin \pi (m-\ell)\alpha|.
\end{equation*}
We claim that 
\begin{equation}\label{sigma+}
    \Sigma_+ \geqslant -s q_{k} \ln 2 - Cs \ln q_{k} + s \ln \frac{s}{q_{k+1}} - C\sigma_{2}\ln q_{n}-\Omega_{2}\sigma_{2}\ln{q_{n_{2}}},
\end{equation}
and
\begin{equation}\label{sigma-}
    \Sigma_- \geqslant -s q_{k} \ln 2 - Cs \ln q_{k} + s \ln \frac{s}{q_{k+1}} - C\sigma_{1} \ln q_{n}-\Omega_{1}\sigma_{1}\ln q_{n_{1}}.
\end{equation}

We only prove \eqref{sigma+} since the proof of \eqref{sigma-} is similar. Note that $I$ can be represented as a disjoint union of $s$ segments, denoted by $J_i$ for $i=1, \dots, s$, each of length $q_{k}$. Applying \hyperref[av]{Lemma~\ref*{av}}  to each $J_i$ yields
\begin{equation*}
    \Sigma_+ \geqslant -s q_{k} \ln 2 - Cs \ln q_{k} + \sum_{i=1}^{s} \ln |\sin \pi \hat{\theta}_i|,
\end{equation*}
where 
\begin{equation*}
    |\sin \pi \hat{\theta}_i| \coloneq \min_{\ell \in J_i} |\sin \pi (2\theta + (\ell + m)\alpha)|.
\end{equation*}

To bound $\sum_{i=1}^{s} \ln |\sin \pi \hat{\theta}_i|$, we first consider the case where $\hat{\theta}_{i+1} = \hat{\theta}_i + q_{k} \alpha$ for $1 \leqslant i < s$. In this case,
\begin{equation*}
    \|\hat{\theta}_{i+1} - \hat{\theta}_i\|_{\mathbb{R}/\mathbb{Z}} \geqslant \|q_{k} \alpha\|_{\mathbb{R}/\mathbb{Z}} \geqslant \frac{1}{2q_{k+1}}.
\end{equation*}
Denote
\begin{equation*}
    \mathcal{E}=\{i: \ln |\sin \pi \hat{\theta}_{i}|<-\sigma_{2} \ln q_{n}-C\}.
\end{equation*}
Then $\# \mathcal{E}\leqslant \Omega_{2}$. For $i\in\mathcal{E}$, we use the assumption
\begin{equation*}
    \ln|\sin \pi \hat{\theta}_{i}| \geqslant -\sigma_{2} \ln q_{n_{2}}-C,
\end{equation*}
and thus
\begin{equation}\label{E}
    \sum_{i\in\mathcal{E}} \ln |\sin \pi \hat{\theta}_i| \geqslant  -\Omega_{2}\sigma_{2}\ln q_{n_{2}}-C \Omega_{2}.
\end{equation}
For $i\notin \mathcal{E}$, we have
\begin{equation*}
    \ln|\sin \pi \hat{\theta}_{i}| \geqslant -\sigma_{2} \ln q_{n}-C.
\end{equation*}
Applying Stirling's formula, we obtain
\begin{equation}\label{Ec}
    \begin{split}
        \sum_{i\notin \mathcal{E}} \ln |\sin \pi \hat{\theta}_i| 
        &\geqslant  \sum_{j=1}^{s-1-\#\mathcal{E}} \ln \|j q_{k} \alpha\|_{\mathbb{R}/\mathbb{Z}} -\sigma_{2}\ln q_{n}-C\\
        &\geqslant s \ln \frac{s}{q_{k+1}} - Cs -\sigma_{2}\ln q_{n}.
    \end{split}
\end{equation}
By \eqref{E} and \eqref{Ec}, we have
\begin{equation*}
    \sum_{i=1}^{s}\ln |\sin \pi \hat{\theta}_i|  \geqslant s \ln \frac{s}{q_{k+1}} - Cs -\sigma_{2}\ln q_{n}-\Omega_{2}\sigma_{2}\ln q_{n_{2}}.
\end{equation*}

For the other cases, we decompose $[1, s]$ into maximal intervals $T_\kappa$ such that $\hat{\theta}_{i+1} = \hat{\theta}_i + q_{k} \alpha$ for all $i, i+1 \in T_\kappa$. Then $T_{\kappa}\neq [1,s]$. The boundary points of each $T_\kappa$ satisfy either  are the boundary  of $[1, s]$ or satisfy the inequality $\|\hat{\theta}_i\|_{\mathbb{R}/\mathbb{Z}} + \|q_{k} \alpha\|_{\mathbb{R}/\mathbb{Z}} \geqslant \frac{1}{2} \|q_{k-1} \alpha\|_{\mathbb{R}/\mathbb{Z}}$. This arises from the fact that there exists $0 < |z| < q_{k}$ such that
\begin{equation*}
    \|\hat{\theta}_i + (z + q_{k}) \alpha\|_{\mathbb{R}/\mathbb{Z}}< \|\hat{\theta}_i +  q_{k}\alpha\|_{\mathbb{R}/\mathbb{Z}},
\end{equation*}
and for such $z$,
\begin{equation*}
    \|\hat{\theta}_i + (z + q_{k}) \alpha\|_{\mathbb{R}/\mathbb{Z}} \geqslant \|q_{k-1} \alpha\|_{\mathbb{R}/\mathbb{Z}} - \|\hat{\theta}_i\|_{\mathbb{R}/\mathbb{Z}} - \|q_{k} \alpha\|_{\mathbb{R}/\mathbb{Z}}.
\end{equation*}
If $T_\kappa$ contains some $i$ such that $\|\hat{\theta}_i\|_{\mathbb{R}/\mathbb{Z}} < \frac{1}{4} \|q_{k-1} \alpha\|_{\mathbb{R}/\mathbb{Z}}$, then
\begin{equation*}
    \# T_\kappa \geqslant \frac{1}{4} \frac{\|q_{k-1} \alpha\|_{\mathbb{R}/\mathbb{Z}}}{\|q_{k} \alpha\|_{\mathbb{R}/\mathbb{Z}}} - 1 \geqslant \frac{1}{8} \frac{q_{k+1}}{q_{k}} - 1\geqslant \frac{s}{8C_{0}''}-1.
\end{equation*}
Combine $\# T_{\kappa}\geqslant 1$, we have $\# T_{\kappa}\geqslant \max \{1, s/8C_{0}''-1\}\geqslant \frac{s}{16C_{0}''}$.
Moreover, 
\begin{equation*}
    \#\bigg\{\kappa: T_{\kappa} \ \text{contains}\ i \ \text{with}\  \|\hat{\theta}_i\|_{\mathbb{R}/\mathbb{Z}} < \frac{1}{4} \|q_{k-1} \alpha\|_{\mathbb{R}/\mathbb{Z}}\bigg\}\leqslant 16C_{0}''.
\end{equation*}
For such $T_\kappa$, an estimate similar to the above gives
\begin{equation}\label{case01}
    \begin{split}
        \sum_{i \in T_\kappa} \ln |\sin \pi \hat{\theta}_i| &\geqslant \# T_\kappa \ln \frac{\# T_\kappa}{q_{k+1}} - C \# T_\kappa - \sigma_{2}\ln{q_n} -\Omega_{2,\kappa} \sigma_{2}\ln q_{n_{2}}\\
        &\geqslant \# T_\kappa \ln \frac{s}{q_{k+1}} - C \# T_\kappa - \sigma_{2}\ln{q_n}-\Omega_{2,\kappa} \sigma_{2}\ln q_{n_{2}},
    \end{split}
\end{equation}
where $\Omega_{2,\kappa}$ denotes the number of resonances of phase in $T_{\kappa}$.

Conversely, if $T_\kappa$ does not contain such $i$, then
\begin{equation*}
    \ln|\sin\pi \hat{\theta}_{i}|\geqslant -\ln q_{k}-C.
\end{equation*}
Since $s q_{k} < C_{0}''q_{k+1}$, we have
\begin{equation}\label{case02}
    \begin{split}
        \sum_{i \in T_\kappa} \ln |\sin \pi \hat{\theta}_i| &\geqslant -\# T_\kappa \ln {q_{k}} - C \# T_\kappa\\
        &\geqslant \# T_\kappa \ln \frac{s}{q_{k+1}} - C \# T_\kappa.
    \end{split}
\end{equation}
Combining \eqref{case01} with \eqref{case02} and $\sum_{\kappa} \Omega_{2,\kappa}\leqslant \Omega_{2}$, we obtain 
\begin{equation*}
    \sum_{i=1}^{s} \ln |\sin \pi \hat{\theta}_i| \geqslant s \ln \frac{s}{q_{k+1}} - Cs - C\sigma_{2}\ln q_{n}-\Omega_{2}\sigma_{2}\ln{q_{n_{2}}},
\end{equation*}
which implies
\begin{equation*}
    \Sigma_+ \geqslant -s q_{k} \ln 2 - Cs \ln q_{k} + s \ln \frac{s}{q_{k+1}} - C\sigma_{2}\ln q_{n}-\Omega_{2}\sigma_{2}\ln{q_{n_{2}}}.
\end{equation*}
Combining \eqref{sigma+} with \eqref{sigma-} yields
\begin{equation}\label{den01}
    \begin{split}
        \sum_{\ell \neq m} \ln |\cos 2\pi \theta_m - &\cos 2\pi \theta_\ell| \geqslant -s q_{k} \ln 2 - Cs \ln q_{k} + 2s \ln \frac{s}{q_{k+1}} \\
        &- C(\sigma_{1}+\sigma_{2}) \ln q_n -\Omega_{1}\sigma_{1}\ln q_{n_{1}}-\Omega_{2}\sigma_{2}\ln{q_{n_{2}}}.
    \end{split}
\end{equation}
Finally, by comparing the estimates in \eqref{num01} and \eqref{den01}, we conclude that
\begin{equation*}
    \mathrm{Lag}_m \leqslant -2s \ln \frac{s}{q_{k+1}} + Cs \ln q_{k} + C(\sigma_{1}+\sigma_{2}) \ln q_n+\Omega_{1}\sigma_{1}\ln q_{n_{1}}+\Omega_{2}\sigma_{2}\ln q_{n_{2}}.
\end{equation*}
This proves \hyperref[item: uniformk]{Theorem~\ref*{uniform}\,\textup{(\ref*{item: uniformk})}}.

Let us prove \hyperref[item: uniformn]{Theorem~\ref*{uniform}\,\textup{(\ref*{item: uniformn})}}. In fact, there is no difference for the estimate of numerator. The only difference comes from denominator. Write
\begin{equation*}
    \begin{split}
        \sum_{\ell \neq m} \ln |\cos 2\pi \theta_m - \cos 2\pi \theta_\ell| &= \Sigma_+ + \Sigma_- + (s q_{k} - 1) \ln 2,
    \end{split}
\end{equation*}
where
\begin{equation*}
    \Sigma_+ = \sum_{\ell \neq m} \ln |\sin \pi (2\theta + (m+\ell)\alpha)|, 
\end{equation*}
and
\begin{equation*}
    \Sigma_- = \sum_{\ell \neq m} \ln |\sin \pi (m-\ell)\alpha|.
\end{equation*}
We have
\begin{equation}\label{Sigma2}
    \Sigma_+ \geqslant -s q_{k} \ln 2 - Cs \ln q_{k} + \sum_{i=1}^{s} \ln |\sin \pi \hat{\theta}_i|,
\end{equation}
where 
\begin{equation*}
    |\sin \pi \hat{\theta}_i| \coloneq \min_{\ell \in J_i} |\sin \pi (2\theta + (\ell + m)\alpha)|.
\end{equation*}

To bound $\sum_{i=1}^{s} \ln |\sin \pi \hat{\theta}_i|$, we first consider the case where $\hat{\theta}_{i+1} = \hat{\theta}_i + q_{k} \alpha$ for $1 \leqslant i < s$. By Stirling's formula and $q_{k+1}\leqslant q_{n}$, we obtain
\begin{equation}\label{wholeline}
    \begin{split}
        \sum_{i=1}^{s}\ln |\sin \pi \hat{\theta}_i|  &\geqslant s \ln \frac{s}{q_{k+1}} - Cs -\sigma_{2}\ln q_{n}-\Omega_{2}\sigma_{2}\ln q_{n_{2}}\\
        &\geqslant s \ln \frac{s}{q_{n}} - Cs -\sigma_{2}\ln q_{n}-\Omega_{2}\sigma_{2}\ln q_{n_{2}}.
    \end{split}
\end{equation}
For the other cases, we decompose $[1, s]$ into maximal intervals $T_\kappa$ such that $\hat{\theta}_{i+1} = \hat{\theta}_i + q_{k} \alpha$ for all $i, i+1 \in T_\kappa$.

If $T_\kappa$ contains some $i$ such that $\|\hat{\theta}_i\|_{\mathbb{R}/\mathbb{Z}} < \frac{1}{4} \|q_{k-1} \alpha\|_{\mathbb{R}/\mathbb{Z}}$, then
\begin{equation*}
    \# T_\kappa \geqslant \frac{1}{4} \frac{\|q_{k-1} \alpha\|_{\mathbb{R}/\mathbb{Z}}}{\|q_{k} \alpha\|_{\mathbb{R}/\mathbb{Z}}} - 1 \geqslant \frac{1}{8} \frac{q_{k+1}}{q_{k}} - 1.
\end{equation*}
Combine $\# T_{\kappa}\geqslant 1$, we have 
\begin{equation*}
    \# T_{\kappa}\geqslant \max \bigg\{1, \frac{1}{8} \frac{q_{k+1}}{q_{k}} - 1\bigg\}\geqslant \frac{q_{k+1}}{16 q_{k}}.
\end{equation*}
Moreover, 
\begin{equation*}
    \#\bigg\{\kappa: T_{\kappa} \ \text{contains}\ i \ \text{with}\  \|\hat{\theta}_i\|_{\mathbb{R}/\mathbb{Z}} < \frac{1}{4} \|q_{k-1} \alpha\|_{\mathbb{R}/\mathbb{Z}}\bigg\}\leqslant \frac{16sq_{k}}{q_{k+1}}.
\end{equation*}
For such $T_\kappa$, 
\begin{equation}\label{case21}
    \begin{split}
        \sum_{i \in T_\kappa} \ln |\sin \pi \hat{\theta}_i| &\geqslant \# T_\kappa \ln \frac{\# T_\kappa}{q_{k+1}} - C \# T_\kappa - \sigma_{2}\ln{q_n} -\Omega_{2,\kappa} \sigma_{2}\ln q_{n_{2}}\\
        &\geqslant \# T_\kappa \ln \frac{1}{q_{k}} - C \# T_\kappa - \sigma_{2}\ln{q_n}-\Omega_{2,\kappa} \sigma_{2}\ln q_{n_{2}},
    \end{split}
\end{equation}
where $\Omega_{2,\kappa}$ denotes the number of resonances of phase in $T_{\kappa}$.

Conversely, if $T_\kappa$ does not contain such $i$, then
\begin{equation*}
    \ln|\sin\pi \hat{\theta}_{i}|\geqslant -\ln q_{k}-C,
\end{equation*}
and thus
\begin{equation}\label{case22}
    \sum_{i \in T_\kappa} \ln |\sin \pi \hat{\theta}_i| \geqslant -\# T_\kappa \ln {q_{k}} - C \# T_\kappa.
\end{equation}
By \eqref{case21} and \eqref{case22},
\begin{equation}\label{decom2}
    \sum_{i=1}^{s} \ln |\sin \pi \hat{\theta}_i| \geqslant -s \ln q_{k} - C\sigma_{2}\frac{sq_{k}}{q_{k+1}}\ln q_{n}-\Omega_{2}\sigma_{2}\ln{q_{n_{2}}},
\end{equation}
Combine \eqref{decom2} with \eqref{wholeline}, one has
\begin{equation*}
    \sum_{i=1}^{s} \ln |\sin \pi \hat{\theta}_i| \geqslant s \ln \frac{s}{q_{n}}-s \ln q_{k} - C\sigma_{2}(1+\frac{sq_{k}}{q_{k+1}})\ln q_{n}-\Omega_{2}\sigma_{2}\ln{q_{n_{2}}},
\end{equation*}
which, by \eqref{Sigma2}, implies
\begin{equation*}
    \Sigma_+ \geqslant s \ln \frac{s}{q_{n}} -s q_{k} \ln 2 - Cs \ln q_{k} - C\sigma_{2}(1+\frac{sq_{k}}{q_{k+1}})\ln q_{n}-\Omega_{2}\sigma_{2}\ln{q_{n_{2}}}.
\end{equation*}
Similarly,
\begin{equation*}
    \Sigma_{-} \geqslant s \ln \frac{s}{q_{n}} -s q_{k} \ln 2 - Cs \ln q_{k} - C\sigma_{1}(1+\frac{sq_{k}}{q_{k+1}})\ln q_{n}-\Omega_{1}\sigma_{1}\ln{q_{n_{1}}}.
\end{equation*}
This leads to
\begin{equation*}
    \begin{split}
        \Lag_{m}&\leqslant -2s \ln \frac{s}{q_{n}}+Cs \ln q_{k} + C(\sigma_{2}+\sigma_{1})(1+\frac{sq_{k}}{q_{k+1}})\ln q_{n}\\
        &\qquad +\Omega_{1}\sigma_{1}\ln{q_{n_{1}}}+\Omega_{2}\sigma_{2}\ln{q_{n_{2}}}.
    \end{split}
\end{equation*}
\end{proof}

\section*{Acknowledgments}
Jiawei He was supported by the National Natural Science Foundation of China (Grant No. 12501247), the Natural Science Foundation of Fujian Province (Grant No. 2025J08232), the  Startup Fund for Advanced Talents of Putian University (Grant No. 2023120), and the Fujian Alliance of Mathematics (Grant No. 2025SXLMMS07).

\section*{Statements and Declarations}
{\bf Conflict of Interest} 
The authors declare no conflicts of interest.
				
\vspace{0.2in}
{\bf Data Availability}
Data sharing is not applicable to this article as no new data were created or analyzed in this study.

\bibliography{main}
\end{document}